\documentclass[journal,letter]{IEEEtranTCOM}

\normalsize

\ifCLASSINFOpdf
\else
\fi

\usepackage{times} 
\usepackage{epsfig}
\usepackage{grffile}
\usepackage{amssymb}
\usepackage{amsmath}
\usepackage{slashbox}
\usepackage{stmaryrd}
\usepackage{flushend}

\newcommand{\argmax}{\operatornamewithlimits{argmax}}
\newcommand{\argmin}{\operatornamewithlimits{argmin}}
\newtheorem{theorem}{Theorem}[section]
\newtheorem{lemma}{Lemma}[section]  
\newtheorem{conjecture}{Conjecture}[section]
\newcommand{\reffig}[1]{Figure~\ref{#1}}
\newcommand{\reffigwo}[1]{\ref{#1}}

\newcommand{\refsec}[1]{Section~\ref{#1}}
\newcommand{\refthe}[1]{Theorem~\ref{#1}}
\newcommand{\reflem}[1]{Lemma~\ref{#1}}

\newcommand{\refeqwo}[1]{(\ref{#1})}

\newcommand{\bs}[1]{{\boldsymbol{#1}}}
\newcommand{\sym}{{{\psi}}}

\def\cE{{\mathcal E}}

\def\cH{{\mathcal H}}
\def\cI{{\mathcal I}}

\def\b0{{\bs 0}}

\newcounter{mycite}
\newtoks\citetoks
\makeatletter
\DeclareRobustCommand\unscite[1]{%
  \@ifundefined{uns@cite#1}
    {\refstepcounter{mycite}\label{citelabel@#1}%
     \expandafter\xdef\csname uns@cite#1\endcsname{\arabic{mycite}}%
     \toks\z@=\expandafter{\the\citetoks}%
     \toks\tw@=\expandafter\expandafter\expandafter{%
       \csname uns@bibitem#1\endcsname}%
     \edef\@tempcite{\the\toks\z@\the\toks\tw@}%
     \global\citetoks=\expandafter{\@tempcite}%
    }{}[\@nameuse{uns@cite#1}]}
\newcommand{\mybibitem}[2]{%
  \@namedef{uns@bibitem#1}{\bibitem[\ref{citelabel@#1}]{#1}#2}}
\makeatother

\begin{document}
%
\title{Low-Complexity LP Decoding \\ of Nonbinary Linear Codes}
%
%
%

\author{Mayur~Punekar,~\IEEEmembership{Student Member,~IEEE,}
Pascal~O.~Vontobel,~\IEEEmembership{Senior Member,~IEEE,} \\
and Mark~F.~Flanagan,~\IEEEmembership{Senior Member,~IEEE}
\thanks{The work of M.~Punekar and M.~F.~Flanagan was supported by Science Foundation Ireland grants 07/SK/I1252b and 06/MI/006. The material in this paper was presented in part at the 48th Annual Allerton Conference on Communication, Control and Computing, Monticello, Illinois, Sept./Oct. 2010, and in part at the 2011 IEEE International Symposium on Information Theory, St. Petersburg, Russia, Jul./Aug. 2011.}
\thanks{M.~Punekar and M.~F.~Flanagan are with the Claude Shannon Institute, University College Dublin, Ireland (e-mail: \{mayur.punekar, mark.flanagan\}@ieee.org).}
\thanks{P.~O.~Vontobel is with Hewlett--Packard Laboratories, 1501 Page Mill Road, Palo Alto, CA 94304, USA (e-mail: pascal.vontobel@ieee.org).}
%
}

\markboth{To appear in IEEE Transactions on Communications}%
{Punekar \MakeLowercase{\textit{et al.}}: Low-complexity LP decoding of nonbinary linear codes}
%



\maketitle

\vspace{-60pt}
\begin{abstract}
Linear Programming (LP) decoding of Low-Density Parity-Check (LDPC) codes has attracted much attention in the research community in the past few years. LP decoding has been derived for binary and nonbinary linear codes. However, the most important problem with LP decoding for both binary and nonbinary linear codes is that the complexity of standard LP solvers such as the simplex algorithm remains prohibitively large for codes of moderate to large block length. To address this problem, two low-complexity LP (LCLP) decoding algorithms for binary linear codes have been proposed by Vontobel and Koetter,
henceforth called the {\em basic LCLP decoding algorithm} and the {\em subgradient LCLP decoding algorithm}. 
In this paper, we generalize these LCLP decoding algorithms to nonbinary linear codes.
The computational complexity per iteration of the proposed nonbinary LCLP decoding algorithms scales linearly with the block length of the code.
A modified BCJR algorithm for efficient check-node calculations in the nonbinary basic LCLP decoding algorithm is also proposed, which has complexity linear in the check node degree. 
Several simulation results are presented for nonbinary LDPC codes defined over $\mathbb{Z}_4$, GF($\mathbf{4}$), and GF($\mathbf{8}$) using quaternary phase-shift keying and 8-phase-shift keying, respectively, over the AWGN channel. It is shown that for some group-structured LDPC codes, the error-correcting performance of the 
nonbinary LCLP decoding algorithms is similar to or better than that of the min-sum decoding algorithm.
\end{abstract}

%
%
%
%
\IEEEpeerreviewmaketitle
%
\section{Introduction}
Low-Density Parity-Check (LDPC) codes have attracted much attention in the research community in the past decade. LDPC codes are generally decoded by message-passing iterative decoding methods such as the sum-product (SP) algorithm, also known as \emph{belief propagation} (BP), and the min-sum (MS) algorithm, which perform remarkably well at moderate SNR levels. 
However, binary LDPC codes often suffer from an \textit{error-floor} effect in the high-SNR region. Some progress has been made in the direction of finite-length analysis of LDPC codes and concepts such as stopping sets \unscite{DiPr_02}, trapping sets \unscite{Ri_03}, graph-cover pseudocodewords \unscite{VoKo_IT}, etc., were introduced and investigated to understand the behavior of the SP algorithm in the error-floor region. Nevertheless, finite-length analysis of LDPC codes under the SP algorithm is a difficult task.

The main focus of research in the area of LDPC codes has been on \textit{binary} LDPC codes. However, it is desirable to use nonbinary LDPC codes in many applications where bandwidth efficient higher order (i.e., nonbinary) modulation schemes are used. Nonbinary LDPC codes are also considered for storage applications \unscite{MaHa_09}. Nonbinary LDPC codes and the corresponding nonbinary SP algorithm were investigated by Davey and MacKay in \unscite{DaMa_98}, and since then many code construction methods and optimized nonbinary SP algorithms have been proposed. However, the finite-length analysis of nonbinary LDPC codes under the nonbinary SP algorithm is also difficult and attempts in this direction (see, e.g., \unscite{AnKa_10}) have been few.

An alternative decoding algorithm for binary LDPC codes, known as linear programming (LP) decoding\footnote{In this paper, the acronym LP stands for {\it linear programming} or {\it linear program}, depending on the context.}, 
was proposed by Feldman \textit{et al.} in \unscite{Fe_Th_03}, \unscite{FeWa_05}. In LP decoding, the ML decoding problem is modeled as an integer programming (IP) problem which is then relaxed to obtain the corresponding LP problem. This LP problem is solved with the help of standard LP solvers based on the simplex algorithm or interior-point methods. Compared to SP decoding, LP decoding relies on the well-studied mathematical theory of LP. Hence, LP decoding is better suited to mathematical analysis and it is possible to make statements about its complexity and convergence, as well as to place bounds on its error-correcting performance. However, the worst-case time complexity of the LP solvers based on the simplex method is known to be exponential in the description complexity, and with other LP solvers based on interior-point methods the corresponding worst-case time complexity is polynomial. On the other hand, iterative decoding algorithms such as the SP algorithm have (per iteration) time complexity linear in the block length of the code and hence significantly outperform LP decoding algorithms based on simplex or interior-point methods in terms of efficiency.

To overcome the complexity problem, several improved LP decoding algorithms have been proposed in 
\unscite{VoKo_06}, \unscite{YaFe_06}, \unscite{YaWa_08}, \unscite{TaSh_11}, \unscite{BaLu_08}, etc. 
In \unscite{VoKo_06} and \unscite{VoKo_07}, the authors use techniques from LP and coding theory to derive two low-complexity LP (LCLP) decoding algorithms, namely the {\it basic LCLP decoding algorithm} and the {\it subgradient LCLP decoding algorithm}, which can be used for approximate LP decoding of binary LDPC codes.
The basic and subgradient LCLP decoding algorithms rely on the {\it block-coordinate ascent method} (also known as the {\it nonlinear Gauss-Seidel method}) \unscite{Be_99} 
and the {\it incremental subgradient algorithm} \unscite{Ne_02}, respectively, to 
obtain a solution to the LP problem proposed in \unscite{FeWa_05}.
Also, the variable node (VN) and check node (CN) calculations of 
the basic LCLP decoding algorithm are directly 
related to VN and CN calculations of the binary SP algorithm; hence
the complexity of each iteration of 
the basic LCLP decoding algorithm 
is similar to that of the SP algorithm. The complexity of each iteration of the subgradient LCLP decoding algorithm is similar to that of the min-sum algorithm.
An algorithm similar to the basic LCLP decoding algorithm 
for more general graphical models was proposed in \unscite{GlJa_07}. An extension of the basic LCLP decoding algorithm 
was proposed and studied in \unscite{Bu_09}. %

In \unscite{FlSk_09}, 
LP decoding was extended from binary linear codes to nonbinary linear codes.
Nonbinary LP decoding, as presented in \unscite{FlSk_09}, relies on standard LP solvers based on simplex or interior-point methods, 
and hence standard iterative decoding algorithms such as the nonbinary SP algorithm significantly outperform these nonbinary LP decoding algorithms in terms of computational complexity. In independent work \unscite{GoBu_10}, \unscite{GoBu_12}, a new scheduling scheme was proposed for the nonbinary basic LCLP decoding algorithm which extends the low-complexity LP decoding method of \unscite{Bu_09} to nonbinary codes.

In this paper we extend the works of \unscite{VoKo_06}, \unscite{VoKo_07} to nonbinary linear codes 
and propose the nonbinary basic and subgradient LCLP decoding algorithms.
We use the LP formulation of nonbinary linear codes proposed in \unscite{FlSk_09} to develop an equivalent primal LP formulation. Then, using the techniques introduced in \unscite{Vo_02} and  \unscite{VoLo_03}, the corresponding dual LP is derived which in turn is used to develop update equations for nonbinary LCLP decoding algorithms. %
The complexity of the proposed nonbinary LCLP decoding algorithms per iteration is linear in the code's block length. In contrast to binary basic LCLP decoding, the VN and CN calculations of nonbinary basic LCLP decoding are not directly related to nonbinary SP. 
Therefore, without the use of an efficient CN processing algorithm, the complexity of the CN calculations will be exponential in the maximum CN degree.
To overcome this problem, we propose a modified BCJR algorithm for efficient CN processing which has complexity linear in the CN degree and allows for efficient implementation of nonbinary basic LCLP decoding. We also propose an alternative state metric which can be used for faster CN processing. %

The remainder of the paper is structured as follows. We begin with some notation and background in Section II. The primal LP is developed in Section III and the corresponding dual LP is given in Section IV. Section V presents the nonbinary basic LCLP decoding algorithm, and reduced complexity CN processing is presented in Section VI. Section VII outlines the nonbinary subgradient LCLP decoding algorithm. Simulation results are presented and discussed in Section VIII. 
%
%

\section{Notation and Background} 
The symbols $\mathbb{R}$, $\mathbb{R}_{>0}$, and $\mathbb{Z}_{>0}$ denote the
field of real numbers, the set of positive real numbers, and the set of
positive integer numbers, respectively. Let $\Re$ be a finite ring with $q$
elements, where $0$ and $1$ denote the additive and multiplicative identity,
respectively, and let $\Re^{-} = \Re \setminus \{0\}$. The standard inner
product of two vectors $\bs{x}$ and $\bs{y}$ of equal length is denoted by
$\left\langle \bs{x}, \bs{y} \right\rangle$.

Let $\mathcal{C}$ be a linear code of length $n$ over the ring $\Re$, defined by
$\mathcal{C} = \{ \bs{c} \in \Re^{n} : \bs{c} \mathcal{H}^{T} = \bs{0} \}$
where $\mathcal{H}$ is an $m \times n$ parity-check matrix with entries from $\Re$. The code $\mathcal{C}$ has rate\footnote{The code rate is defined as the ratio of the number of information symbols to the number of coded symbols. Note that in general, for a code over a ring $\Re$, the code rate may not in general be expressed in terms of the rank of $\mathcal{H}$ (since $\mathcal{H}$ may contain non-invertible elements).} $R(\mathcal{C}) = \log_q (|\mathcal{C}|) / n$ and is referred to as an $[n, \log_q(|\mathcal{C}|)]$ linear code over $\Re$. 

The set $\mathcal{J} = \{1,\ldots,m\}$ denotes row indices and the set $\mathcal{I} = \{1,\ldots,n\}$ denotes column indices of $\mathcal{H}$. We use $\mathcal{H}_j$ for the $j$-th row of $\mathcal{H}$ and $\mathcal{H}^i$ for the $i$-th column of $\mathcal{H}$. The support of the vector $\bs{c}$ is denoted by supp$(\bs{c})$. For each $j \in \mathcal{J}$, let $\mathcal{I}_j = \mbox{supp}(\mathcal{H}_j)$ and for each $i \in \mathcal{I}$, let $\mathcal{J}_i = \mbox{supp}(\mathcal{H}^i)$. Also let $d_j = |\mathcal{I}_j|$ and $d = \max_{j \in \mathcal{J}}\{d_j\}$. We define the set $\mathcal{E} = \{(i,j) \in \mathcal{I} \times \mathcal{J} \; : \; j \in \mathcal{J}, i \in \mathcal{I}_j\} = \{(i,j) \in \mathcal{I} \times \mathcal{J} \; : \; i \in \mathcal{I}, j \in \mathcal{J}_i\}$. Moreover, for each $j \in \mathcal{J}$, we define the local single parity check (SPC) code 
$\mathcal{B}_j = \{(b_i)_{i \in \mathcal{I}_j} \in \Re^{|\mathcal{I}_j|} : \sum_{i \in \mathcal{I}_j} b_i \cdot \mathcal{H}_{j,i} = 0 \}.$
For each $i \in \mathcal{I}$, we denote by $\mathcal{A}_i \subseteq \Re^{|\{0\} \cup \mathcal{J}_i|}$ the repetition code of the appropriate length and indexing. We also use the following notation introduced in 
\unscite{VoKo_06}: for a statement $A$ we have $\llbracket$A$\rrbracket = 0$ if $A$ is true and $\llbracket A \rrbracket = +\infty$ otherwise.
As in \unscite{FlSk_09}, we define the mapping 
\[
\bs{\xi} : \Re \rightarrow \{0,1\} ^{q-1} \subset \mathbb{R}^{q-1}
\]
by
\[
\bs{\xi}(r) = \bs{x} = (x^{(\rho)})_{\rho \in \Re^{-}}
\]
such that for each $\rho \in \Re^{-}$
\begin{eqnarray*}
x^{(\rho)} = \left\{ 
\begin{array}{l l}
  1 & \; \text{if} \; \rho = r \\
  0 & \; \text{otherwise} .\\

\end{array} \right. 
\end{eqnarray*}
Building on this we define
\[
\bs{\Xi} : \underset{\scriptstyle t \in \mathbb{Z}_{>0}}{\cup} \Re^{t} \rightarrow \underset{\scriptstyle t \in \mathbb{Z}_{>0}}{\cup} \{0,1\}^{(q-1)t} \subset \underset{\scriptstyle t \in \mathbb{Z}_{>0}}{\cup} \mathbb{R}^{(q-1)t} \; ,
\]
according to $\bs{\Xi} (\bs{c}) = (\bs{\xi}(c_1),\dots,\bs{\xi}(c_t))\nonumber, \quad \forall \bs{c} \in \Re^{t}, t \in \mathbb{Z}_{>0}.$

For vectors $\bs{f} \in \mathbb{R}^{(q-1)n}$ we use the notation
$\bs{f} = (\bs{f}_1 \;| \; \bs{f}_2 \; | \; \ldots \; | \; \bs{f}_n) \text{ where }
%
\forall i \in \mathcal{I}, \bs{f}_i = (f_{i}^{(r)})_{r \in \Re^{-}} \; .$


%
%
%
%
We also define the inverse of $\bs{\Xi}$ as
$\bs{\Xi}^{-1}(\bs{f}) = (\bs{\xi}^{-1}(\bs{f}_1), \bs{\xi}^{-1}(\bs{f}_2), \ldots, \bs{\xi}^{-1}(\bs{f}_n)).$
Note that the inverse of $\bs{\Xi}$ is well defined for any $\bs{f} \in \mathbb{R}^{(q-1)n}$ where each component $\bs{f}_i$, $i  \in \mathcal{I}$, has entries from $\{ 0,1 \}$ with sum at most $1$. 

We assume data transmission over a $q$-ary input memoryless channel 
whose input alphabet is identified with $\Re$, and whose output alphabet is denoted by $\Sigma$. The received vector is denoted by $\bs{y} =(y_1, y_2, \ldots, y_n)\in \Sigma^{n}$.
Based on this, for each $i \in \cI$ we define a vector $\bs{\lambda}_i = (\lambda^{(r)}_i)_{r \in \Re^{-}}$ 
where, for each $y \in \Sigma$, $r \in \Re^{-}$,
$
\lambda^{(r)}_i = \log \left(\frac{p(y_i|0)}{p(y_i|r)}\right).
$
Here $p(y|c)$ denotes the channel output probability (density) conditioned on the channel input. Based on this, we also define
$\bs{\Lambda} = (\bs{\lambda}_1\;|\;\bs{\lambda}_2\;| \dots |\;\bs{\lambda}_n).$

For $\kappa \in \mathbb{R}_{> 0}$, we define the function $\sym (x) = e^{\kappa x}$, and its inverse $\sym^{-1} (x) = \frac{1}{\kappa} \log (x).$ We will use Forney-style factor graphs (FFGs), also known as normal factor graphs \unscite{Fo_01} to represent the linear programs introduced in this paper. An FFG is a diagram that represents the factorization of a function of several variables. (Note that in this paper FFGs will not represent products of functions, but sums of functions.) For more information on FFGs the reader is referred to \unscite{Fo_01}, \unscite{Vo_02}, \unscite{LO_04}.

%
\begin{figure*}
\begin{minipage}[b]{0.45\linewidth}
\centering
\includegraphics[width=0.95\columnwidth, keepaspectratio]{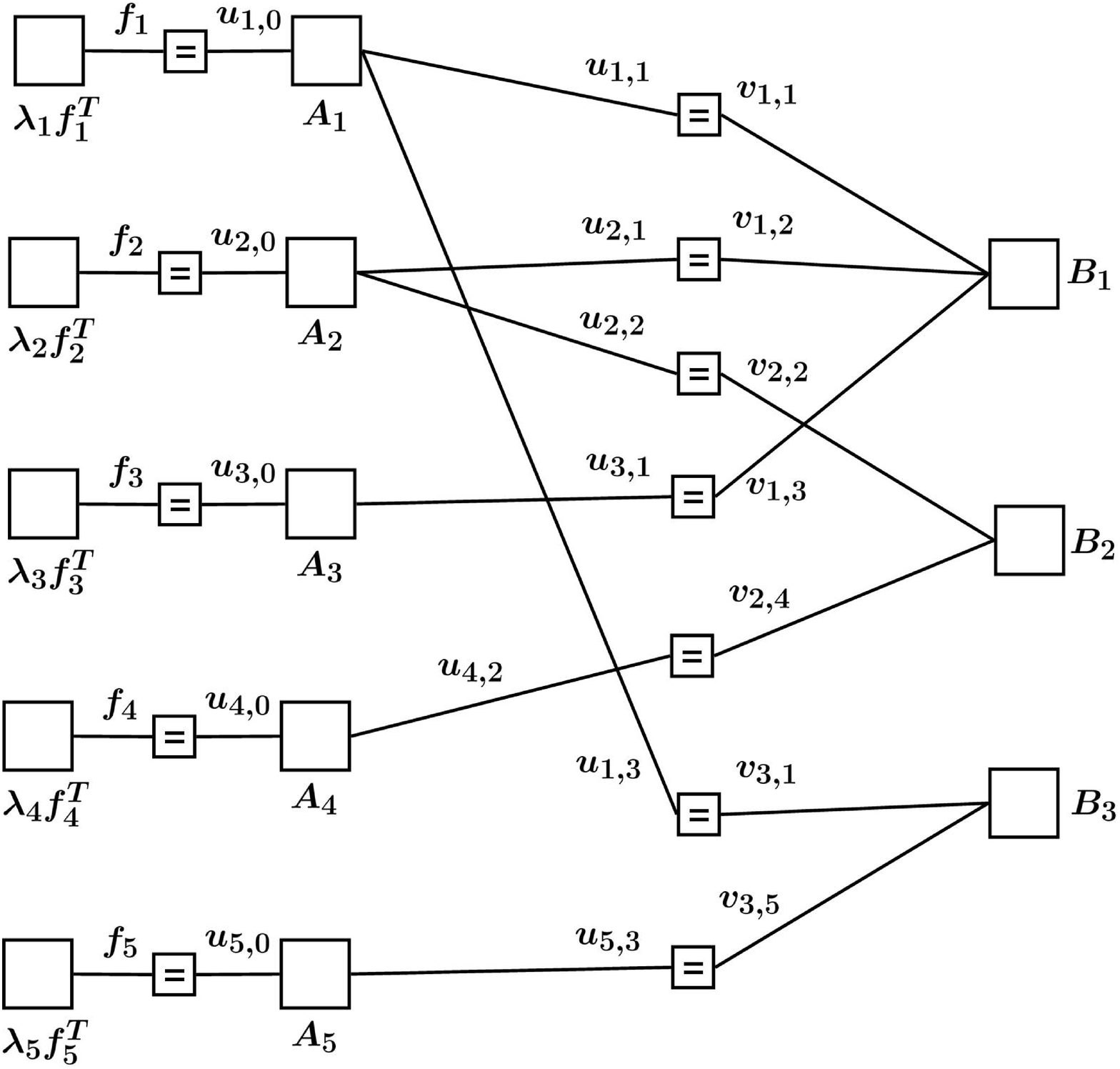}
\caption{FFG which represents the augmented cost function of \eqref{eq:primal} for the example $(5,2)$ nonbinary code.}
\label{fig:primal_LP_FFG}
\end{minipage}
\hspace{1cm}
\setcounter{figure}{2}
\begin{minipage}[b]{0.45\linewidth}
\centering 
\includegraphics[width=0.95\columnwidth, keepaspectratio]{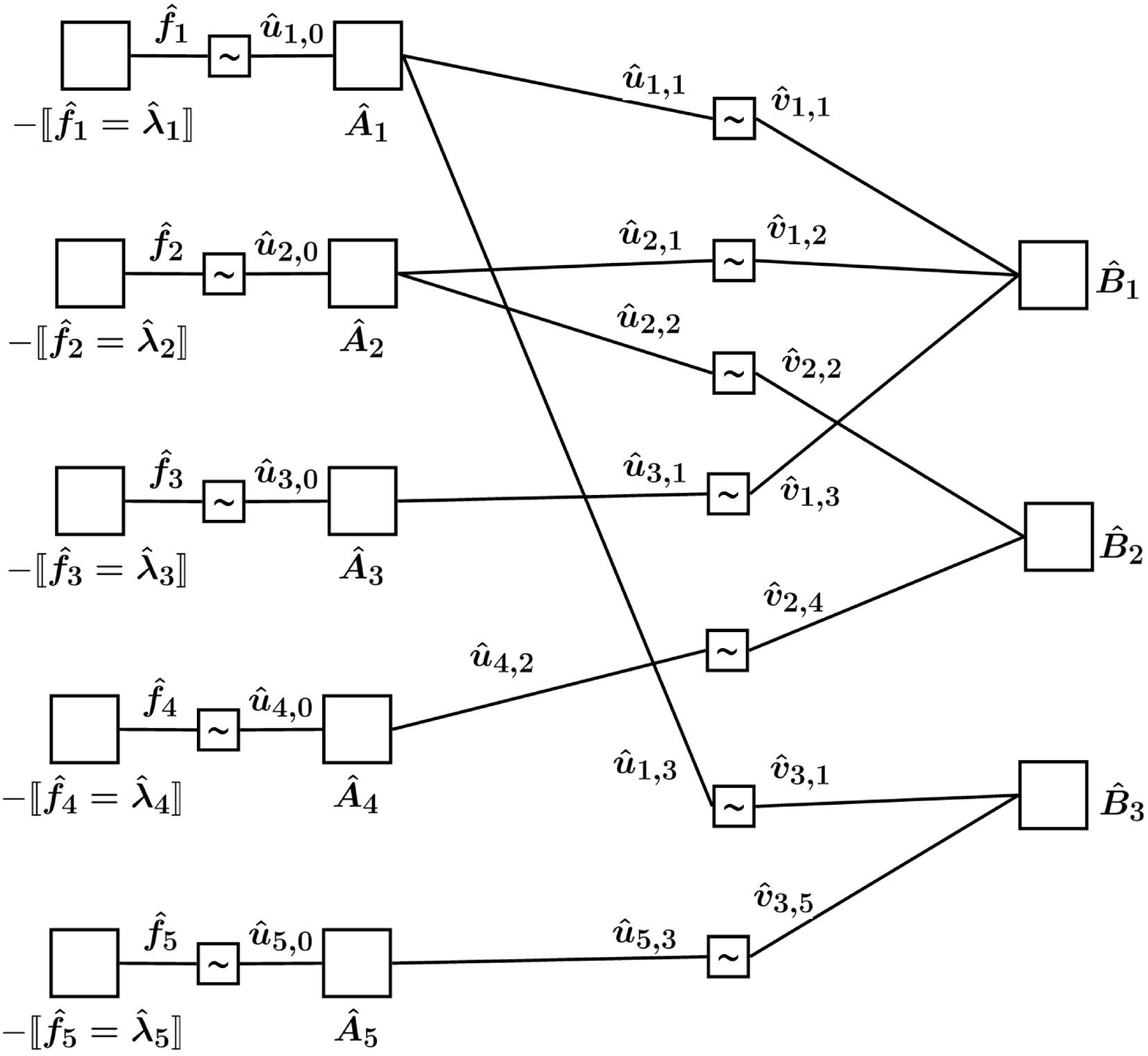}
\caption{FFG which represents the augmented cost function of \eqref{eq:dual} for the example $(5,2)$ nonbinary code. 
Here a function node which is marked with the symbol $\sim$, and which is connected to edges $u$ and $v$, denotes the function $-\llbracket u = -v \rrbracket$.
} 
\label{fig:dual_LP_FFG}
\end{minipage}
\end{figure*}
%
\begin{figure*}
\setcounter{figure}{1}
\begin{minipage}[b]{0.45\linewidth}
\centering
\includegraphics[width=1.0\columnwidth, keepaspectratio]{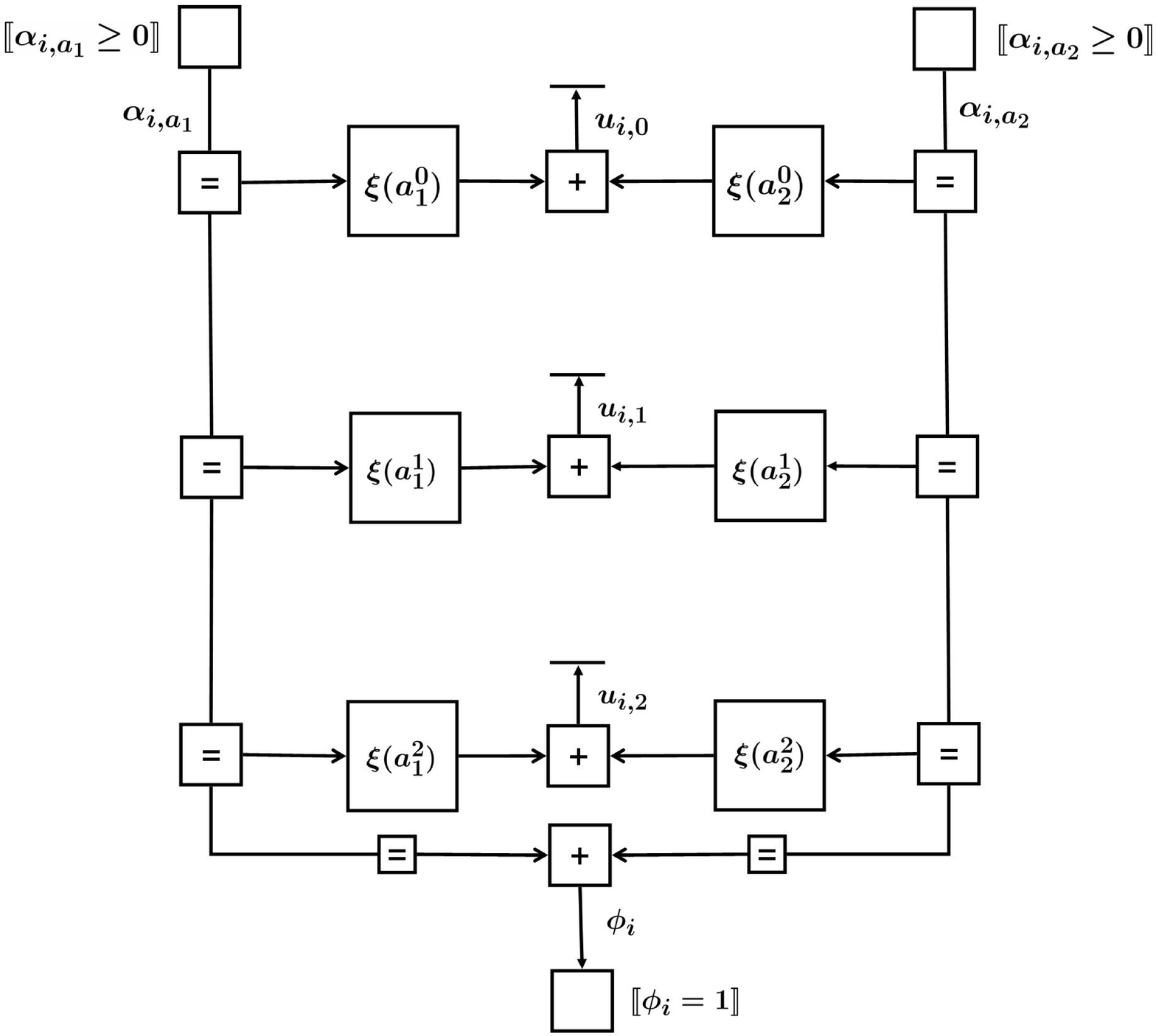}
\caption{FFG for the function $A_i(\bs{u}_i)$. This forms a subgraph of the overall FFG of \reffig{fig:primal_LP_FFG}. The FFG is illustrated for the special case $|\mathcal{A}_i| = 2$ and $|\mathcal{J}_i| = 2$.} 
\label{fig:primal_FFG}
\end{minipage}
\hspace{1cm}
\setcounter{figure}{3}
\begin{minipage}[b]{0.45\linewidth}
\centering 
\includegraphics[width=1.01\columnwidth, keepaspectratio]{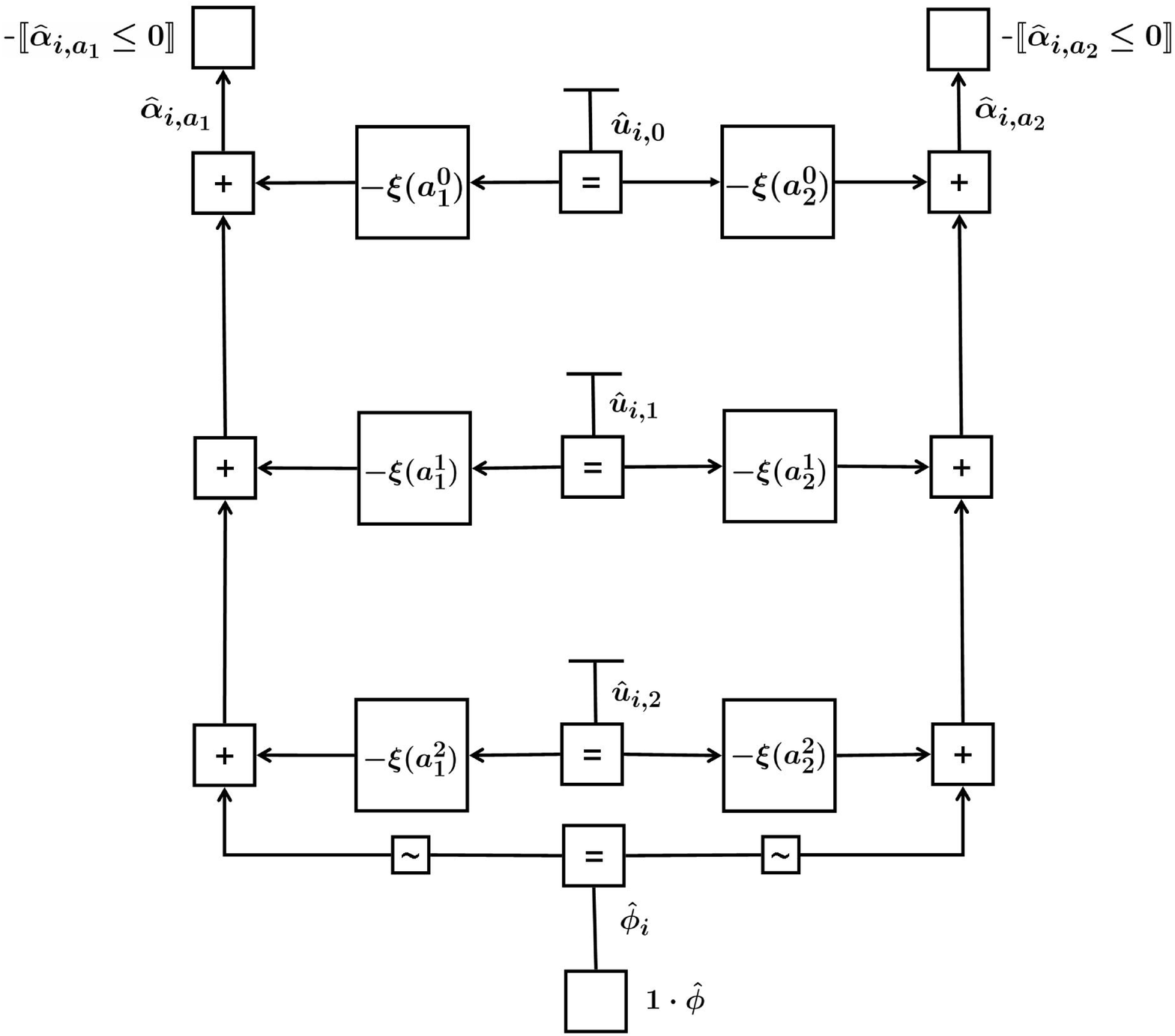}
\caption{FFG for the function $\hat{A}_{i}(\bs{\hat{u}}_{i})$. This FFG is dual to that of \reffig{fig:primal_FFG}. Here, for any primal variable $x$, the dual variable is denoted by $\hat{x}$.} 
\label{fig:dual_FFG}
\end{minipage}
\vspace{-10pt}
\end{figure*}

\section{The Primal Linear Program}
In \unscite{FlSk_09} the authors presented the following linear program to decode nonbinary linear codes: 
%

\vspace{5pt}
\textbf{NBLPD}: 
\begin{align*}
\text{min.} & \quad \sum_{i \in \cI} \bs{\lambda}_i \bs{f}_i^{T} & &\\    
\mbox{subj. to }& &\\
f_{i}^{(r)} &= \sum_{\underset{\scriptstyle b_i = r}{\bs{b} \in \mathcal{B}_j}} w_{j,\bs{b}} & \forall &j \in \mathcal{J}, \; \forall i \in \mathcal{I}_j, \; \forall r \in \Re^{-} \; ,\\
w_{j,\bs{b}} &\ge 0 &\forall &j \in \mathcal{J}, \; \forall \bs{b} \in \mathcal{B}_j \; ,\\
\sum_{\bs{b} \in \mathcal{B}_j} w_{j,\bs{b}} &= 1 &\forall &j \in \mathcal{J} \; .
\end{align*}
We denote the polytope represented by the variables and constraints of {\bf NBLPD} as $\mathcal{Q}_{\mathrm{f}}$. Two alternative polytope representations are also given in \unscite{FlSk_09}, which are both equivalent to {\bf NBLPD}. It is also possible to reformulate the constraints of {\bf NBLPD} with additional auxiliary variables. However, to develop a low-complexity LP decoding algorithm for {\bf NBLPD}, we use the approach of \unscite{VoKo_06} and reformulate {\bf NBLPD} so that the new LP formulation can be directly represented by an FFG:

\medskip
\textbf{PNBLPD}:
\begin{align*}
\text{min.}& \quad \sum_{i \in \cI} \bs{\lambda}_i \bs{f}_i^{T} & &\\
\mbox{subj. to} & & & \nonumber \\
\bs{{f}}_i &= \bs{u}_{i,0} &(&i \in \mathcal{I}),  \nonumber \\
\bs{u}_{i,j} &= \bs{v}_{j,i} &(&(i,j) \in \mathcal{E}), \nonumber\\
\sum_{\bs{a} \in \mathcal{A}_i} \alpha_{i,\bs{a}} \; \bs{\Xi}(\bs{a}) &= \bs{u}_i &(&i \in \mathcal{I}) \; , \nonumber \\
\sum_{\bs{b} \in \mathcal{B}_j} \beta_{j,\bs{b}} \; \bs{\Xi}(\bs{b}) &= \bs{v}_j &(&j \in \mathcal{J}) \; , \nonumber \\
\alpha_{i,\bs{a}} &\geq 0 &(&i \in \mathcal{I}, \bs{a} \in \mathcal{A}_i) \; , \nonumber \\
\beta_{j,\bs{b}} &\geq 0 &(&j \in \mathcal{J}, \bs{b} \in \mathcal{B}_j) \; , \nonumber\\
\sum_{\bs{a} \in \mathcal{A}_i} \alpha_{i,\bs{a}} &= 1 &(&i \in \mathcal{I}) \; , \nonumber \\
\sum_{\bs{b} \in \mathcal{B}_j} \beta_{j,\bs{b}} &= 1 &(&j \in \mathcal{J}) \; . \nonumber
\end{align*}
Here we introduce the definitions $\bs{u}_{i,j} = (u_{i,j}^{(r)})_{r \in \Re^{-}}$ and $\bs{v}_{j,i} = (v_{j,i}^{(r)})_{r \in \Re^{-}}$ for all $i \in \mathcal{I}$, $j \in \mathcal{J}_i \cup \{ 0 \}$. We also define $\bs{u}_i = (\bs{u}_{i,j})_{j \in \mathcal{J}_i \cup \{ 0 \}}$ for $i \in \mathcal{I}$, and $\bs{v}_j = (\bs{v}_{j,i})_{i \in \mathcal{I}_j}$ for $j \in \mathcal{J}$. 
We denote the polytope represented by the variables and constraints of {\bf PNBLPD} by $\mathcal{Q}_{\mathrm{p}}$. It is important to note that along with the convex hulls of the single parity-check codes, {\bf PNBLPD} also explicitly models the convex hulls of the repetition codes. 
The constraints of {\bf NBLPD} and {\bf PNBLPD} appear to be quite different due to the different notations. However, the projection of each polytope onto the variables denoted by $\bs{f}$ is the same in both cases, and therefore the LPs are equivalent from the point of view of decoding. 
\begin{theorem} \label{the:th1}
Polytopes $\mathcal{Q}_{\mathrm{f}}$ and $\mathcal{Q}_{\mathrm{p}}$ are equivalent from an LP decoding perspective, i.e., for every $(\bs{f}, \bs{\alpha}, \bs{\beta}) \in \mathcal{Q}_{\mathrm{p}}$ there exists a $\bs{w}$ such that $(\bs{f}, \bs{w}) \in \mathcal{Q}_{\mathrm{f}}$, and conversely, for every $(\bs{f}, \bs{w}) \in \mathcal{Q}_{\mathrm{f}}$ there exist $\bs{\alpha}, \bs{\beta}$ such that $(\bs{f}, \bs{\alpha}, \bs{\beta}) \in \mathcal{Q}_{\mathrm{p}}$.
\end{theorem}
\begin{proof}
The proof of \refthe{the:th1} can be found in \unscite{PuFl_10}.
\end{proof}
\vspace{5pt}

Before deriving the dual linear program, we reformulate {\bf PNBLPD} so that this LP can be represented by an FFG. For this purpose, the constraints of {\bf PNBLPD} are expressed as additive cost terms (also known as \textit{penalty terms}). The rule for assigning a cost to a configuration of variables is: if a given configuration satisfies the LP constraints then cost $0$ is assigned to this configuration, otherwise $+\infty$ is assigned. The {\bf PNBLPD} is then equivalent to the unconstrained minimization of the augmented cost function
\begin{align}
\sum_{i \in \mathcal{I}} \bs{\lambda}_i \bs{f}_{i}^{T} &+ \sum_{i \in \mathcal{I}} \llbracket \bs{f}_i = \bs{u}_{i,0} \rrbracket + \sum_{(i,j) \in \mathcal{E}} \llbracket \bs{u}_{i,j} = \bs{v}_{j,i} \rrbracket \nonumber \\ 
&+ \sum_{i \in \mathcal{I}} A_i(\bs{u}_i) + \sum_{j \in \mathcal{J}} B_j(\bs{v}_j) \; , \label{eq:primal}
\end{align}
where $\forall i \in \mathcal{I}$ and $\forall j \in \mathcal{J}$ we have defined
\begin{align*}
A_i(\bs{u}_i) &\triangleq \left \llbracket \sum_{\bs{a} \in \mathcal{A}_i} \alpha_{i, \bs{a}} \; \bs{\Xi}(\bs{a}) = \bs{u}_i \right \rrbracket + \sum_{\bs{a} \in \mathcal{A}_i} \llbracket \alpha_{i, \bs{a}} \ge 0 \rrbracket \\ &
+ \left \llbracket \sum_{\bs{a} \in \mathcal{A}_i} \alpha_{i, \bs{a}} = 1 \right \rrbracket , \\
B_j(\bs{v}_j) &\triangleq \left \llbracket \sum_{\bs{b} \in \mathcal{B}_j} \beta_{j, \bs{b}} \; \bs{\Xi}(\bs{b}) = \bs{v}_j \right \rrbracket + \sum_{\bs{b} \in \mathcal{B}_j} \llbracket \beta_{j, \bs{b}} \ge 0 \rrbracket \\ &
+ \left \llbracket \sum_{\bs{b} \in \mathcal{B}_j} \beta_{j, \bs{b}} = 1 \right \rrbracket.
\end{align*}
For ease of illustration we consider a $(5,2)$ code over $\mathbb{Z}_4$ with parity-check matrix 
\begin{eqnarray*}
\mathcal{H} = \left[ 
\begin{array}{l l l l l}
1 & 3 & 1 & 0 & 0 \\
0 & 1 & 0 & 1 & 0 \\
3 & 0 & 0 & 0 & 1 \\
\end{array}
\right] \; .
\end{eqnarray*}
The augmented cost function for this code is represented by the FFG of \reffig{fig:primal_LP_FFG}.

\section{The Dual Linear Program}
In this section we derive the dual LP for {\bf PNBLPD}. 
As shown in subsequent sections, the dual LP is useful for the development of the nonbinary LCLP decoding algorithms.

The dual LP of {\bf PNBLPD} can be derived from the augmented cost function of \eqref{eq:primal}. First we derive the duals of $A_i(\bs{u}_i)$ and $B_j(\bs{v}_j)$. The (primal) FFG of $A_i(\bs{u}_i)$ is shown in \reffig{fig:primal_FFG} and its dual is shown in \reffig{fig:dual_FFG}. For simplicity of exposition, these graphs are shown for the special case $\mathcal{A}_i = \{\bs{a}_1, \bs{a}_2\} = \{(a_1^0, a_1^1, a_1^2), (a_2^0, a_2^1, a_2^2)\}$, i.e., $|\mathcal{A}_i| = 2$ and $|\mathcal{J}_i| = 2$; the corresponding graphs for the general case have a similar structure. The dual FFG is derived with the help of techniques introduced in \unscite{Vo_02} and \unscite{VoLo_03}. 
The dual function $\hat{A}_{i}(\bs{\hat{u}}_{i})$ is obtained from the dual FFG of \reffig{fig:dual_FFG} as
\begin{equation}
\hat{A}_{i}(\bs{\hat{u}}_{i}) = \hat{\phi}_{i} - \sum_{\bs{a} \in \mathcal{A}_i} \llbracket \hat{\alpha}_{i,\bs{a}} \leq  0 \rrbracket \label{eq:dual_1}
\end{equation}
where, because for each $\bs{a} \in \mathcal{A}_i$ it holds that 
%
%
\begin{equation*}
\hat{\alpha}_{i,\bs{a}} = -\hat{\phi}_{i} + \langle -\bs{\hat{u}}_{i}, \bs{\Xi}(\bs{a}) \rangle \; , \nonumber
\end{equation*}
it follows that
\begin{equation}
-\left\llbracket \hat{\alpha}_{i,\bs{a}} \geq 0 \right\rrbracket = -\left\llbracket \hat{\phi}_{i} \le \langle -\bs{\hat{u}}_{i}, \bs{\Xi}(\bs{a}) \rangle \; \right\rrbracket . \label{eq:dual_2}
\end{equation}
From \refeqwo{eq:dual_1} and \refeqwo{eq:dual_2} we obtain
\begin{eqnarray*}
\hat{A}_{i}(\bs{\hat{u}}_{i}) & = & \hat{\phi}_{i} -  \sum_{\bs{a} \in \mathcal{A}_i} \llbracket \hat{\phi}_{i} \le \langle -\bs{\hat{u}}_{i}, \bs{\Xi}(\bs{a}) \rangle \rrbracket \nonumber \\ 
& = & \hat{\phi}_{i} - \left \llbracket  \hat{\phi}_{i} \le \min_{\bs{a} \in \mathcal{A}_i} \langle -\bs{\hat{u}}_{i}, \bs{\Xi}(\bs{a}) \rangle \right\rrbracket .
\end{eqnarray*}
%
%
%

The same procedure can be used to derive the dual of $B_j(\bs{v}_j)$ as
\begin{eqnarray*}
\hat{B}_{j}(\bs{\hat{v}}_{j}) = \hat{\theta}_{j} - \left \llbracket  \hat{\theta}_{j} \le \min_{\bs{b} \in \mathcal{B}_j} \langle -\bs{\hat{v}}_{j},  \bs{\Xi}(\bs{b}) \rangle \right\rrbracket \; .
\end{eqnarray*}
Finally, we use techniques from \unscite{Vo_02}, \unscite{VoLo_03} to derive the dual of the LP that is represented by the FFG in \reffig{fig:primal_LP_FFG}. The resulting LP is a maximization problem that is represented by the FFG in \reffig{fig:dual_LP_FFG}; its cost function equals
\begin{align}
\sum_{i \in \mathcal{I}} \hat{A}_{i}(\bs{\hat{u}}_{i}) &+ \sum_{j \in \mathcal{J}} \hat{B}_{j}(\bs{\hat{v}}_{j}) - \sum_{i \in \mathcal{I}} \llbracket \bs{\hat{f}}_{i} = -\bs{\hat{u}}_{i,0} \rrbracket \nonumber \\ & 
- \sum_{(i,j) \in \mathcal{E}} \llbracket \bs{\hat{u}}_{i,j} = -\bs{\hat{v}}_{j,i} \rrbracket - \sum_{i \in \mathcal{I}} \left\llbracket \bs{\hat{f}}_{i} = - \bs{\lambda}_{i}\right\rrbracket . \label{eq:dual} 
\end{align}
The dual of {\bf PNBLPD} can therefore be deduced as follows.

\medskip
\textbf{DNBLPD}:
\begin{align}
\text{max.}\quad &\sum_{i \in \mathcal{I}} \hat{\phi}_{i} + \sum_{j \in \mathcal{J}} \hat{\theta}_{j} & &\nonumber  \\
\mbox{subj. to } & & &\nonumber \\
&\hat{\phi}_{i} \leq \min_{\bs{a} \in \mathcal{A}_i}\left\langle-\bs{\hat{u}}_i, \bs{\Xi}(\bs{a})\right\rangle &(&i \in \mathcal{I}) \; , \nonumber \\
&\hat{\theta}_{j} \leq \min_{\bs{b} \in \mathcal{B}_j}\left\langle-\bs{\hat{v}}_j, \bs{\Xi}(\bs{b})\right\rangle &(&j \in \mathcal{J}) \; , \nonumber \\
&\bs{\hat{u}}_{i,j} = - \bs{\hat{v}}_{j,i} &(&(i,j) \in \mathcal{E}) \; , \nonumber \\
&\bs{\hat{u}}_{i,0} = - \bs{\hat{f}}_i &(&i \in \mathcal{I}) \; , \nonumber \\
&\bs{\hat{f}}_i = \bs{\lambda}_i &(&i \in \mathcal{I}) \; . \nonumber 
\end{align}
The augmented cost function of \eqref{eq:dual} for the $(5,2)$ binary code is represented by the FFG of \reffig{fig:dual_LP_FFG}. 

We make use of the soft-minimum operator introduced in \unscite{VoKo_06} and derive the {\it softened dual linear program}. For any $\kappa \in \mathbb{R}_{>0}$, the soft-minimum operator is defined as 
\begin{equation*}
\min_{l}{}^{\!\!(\kappa)} \{ z_l \} \triangleq -\frac{1}{\kappa} \log \left( \sum_{l} e^{-\kappa z_l}\right) = -\sym^{-1}\left(\sum_{l}\sym\Big({-z_l}\Big)\right).
\end{equation*}
Note that $\min_{l}{}^{\!\!(\kappa)} \{ z_l \} \le \min_{l} \{z_{l}\}$, with equality attained in the limit as $\kappa \to \infty$.  With this we define the softened dual linear program {\bf SDNBLPD} which is the same as {\bf DNBLPD} except that $\min$ is replaced by $\min {}^{\!\!(\kappa)}$.
%
%
%
%
\section{Nonbinary Basic Low-Complexity Linear Programming Decoding Algorithm} \label{sec:lclp}
As mentioned earlier, the basic LCLP decoding algorithm proposed in
\unscite{VoKo_06} is a block-coordinate ascent type algorithm. The
block-coordinate ascent algorithm iteratively finds the optimum of a given
continuously differentiable function. Each iteration of the block-coordinate
ascent algorithm consists of multiple steps and during each step, a block of
variables (that might also consist of a single variable) is updated so that
the given function is optimized with respect to them, while at the same time
the rest of the variables are kept constant. An iteration of the
block-coordinate ascent algorithm is completed when all variables are updated.

In this section, we derive the nonbinary basic LCLP decoding algorithm. For this, it is important to observe from {\bf SDNBLPD} that the variables $\bs{\hat{u}}_{i,j}$ and $\bs{\hat{v}}_{j,i}$ are {\it coupled} with each other, i.e., we always have $\bs{\hat{u}}_{i,j} = -\bs{\hat{v}}_{j,i}$ for all $(i, j) \in \cE$.

It can be observed that in {\bf SDNBLPD}, $\hat{\phi}_i$ and $\hat{\theta}_j$
are each involved in only one inequality and hence we can replace these
inequalities with equality without changing the optimal solution (the same is
true for {\bf DNBLPD}). With this, let us select an edge $(i,j) \in
\mathcal{E}$ and a ring element $r \in \Re^{-}$, and let us assume that all
variables except $\hat{u}^{(r)}_{i,j}$ are kept constant; then optimizing the
cost of {\bf SDNBLPD} with respect to $\hat{u}^{(r)}_{i,j}$ is equivalent to
optimizing $\hat{h} \left(\hat{u}^{(r)}_{i,j}\right)$, where
%
\begin{eqnarray} \label{eq:local_fun}
\hat{h} \left(\hat{u}^{(r)}_{i,j}\right) \triangleq \min_{\bs{a} \in \mathcal{A}_i} {}^{\!\!(\kappa)} \left\langle -\bs{\hat{u}}_i, \bs{\Xi}(\bs{a})\right\rangle + \min_{\bs{b} \in \mathcal{B}_j} {}^{\!\!(\kappa)} \left\langle-\bs{\hat{v}}_j, \bs{\Xi}(\bs{b})\right\rangle . 
\end{eqnarray}
Although the soft-minimum operator is an approximation of the minimum operator, its advantage lies in ensuring the convexity and differentiability of the function $\hat{h} \left(\hat{u}^{(r)}_{i,j}\right)$ in \eqref{eq:local_fun}, which makes possible the proofs of Lemmas \ref{lm:lclp} and \ref{le:convergence} described below.
%
%
%
%
%

If the current values of the variables $\hat{u}^{(r)}_{i,j}, \hat{\phi}_{i},
\hat{\theta}_{j}$ related to the edge $(i,j) \in \mathcal{E}$ and the ring
element $r \in \Re^{-}$ are replaced with the new values (at the same time
keeping the other variables constant) such that $\hat{h}
\left(\hat{u}^{(r)}_{i,j}\right)$ is maximized, then we can guarantee that the
dual function also increases or else remains constant at its current
value. The new value $\hat{u}^{*(r)}_{i,j}$, which maximizes $\hat{h}
\left(\hat{u}^{(r)}_{i,j}\right)$ is given by
\begin{align} \label{eq:new_u_alpha}
\hat{u}^{*(r)}_{i,j} \triangleq \argmax_{\hat{u}^{(r)}_{i,j}} \;\hat{h} \left(\hat{u}^{(r)}_{i,j}\right) .
\end{align}
Once we have calculated $\hat{u}^{*(r)}_{i,j}$, we can update the variables $\hat{\phi}_i$ and $\hat{\theta}_j$ accordingly. The calculation of $\hat{u}^{*(r)}_{i,j}$ is given in the following lemma.
\begin{lemma} \label{lm:lclp} The value of $\hat{u}^{*(r)}_{i,j}$ of
  \eqref{eq:new_u_alpha} can be calculated using
\begin{equation*}
\quad \hat{u}^{*(r)}_{i,j} = \frac{1}{2} \left( (V_{i,j,\bar{r}} - V_{i,j,r}) - (C_{j,i,\bar{r}} - C_{j,i,r}) \right),
\end{equation*}
where
\begin{align*}
&V_{i,j,\bar{r}} \triangleq - \min_{\underset{\scriptstyle a_j \ne r}{\bs{a} \in \mathcal{A}_i}} {}^{\!\!(\kappa)} \big\langle\! -\bs{\hat{u}}_i, \bs{\Xi}({\bs{a}}) \big\rangle, \\
&V_{i,j,r} \triangleq - \min_{\underset{\scriptstyle a_j = r}{\bs{a} \in \mathcal{A}_i}} {}^{\!\!(\kappa)} \big\langle\! -\tilde{\bs{u}}_i, \bs{\Xi}(\tilde{\bs{a}}) \big\rangle, \\&
C_{j,i,\bar{r}} \triangleq - \min_{\underset{\scriptstyle b_i \ne r}{\bs{b} \in \mathcal{B}_j}} {}^{\!\!(\kappa)} \big\langle\! -\bs{\hat{v}}_j, \bs{\Xi}(\bs{b}) \big\rangle, \\&
C_{j,i,r} \triangleq - \min_{\underset{\scriptstyle b_i = r}{\bs{b} \in
    \mathcal{B}_j}} {}^{\!\!(\kappa)} \big\langle\! -\tilde{\bs{v}}_j, \bs{\Xi}(\tilde{\bs{b}}) \big\rangle. 
\end{align*}
Here the vectors $\tilde{\bs{u}}_{i}$ and $\tilde{\bs{a}}$ are the vectors $\bs{\hat{u}}_i$ and ${\bs{a}}$, respectively, where the $j$-th position is excluded. Similarly, the vectors $\tilde{\bs{v}}_{j}$ and $\tilde{\bs{b}}$ are obtained by excluding the $i$-th position from $\bs{\hat{v}}_j$ and ${\bs{b}}$, respectively.\end{lemma}
\begin{proof}
The proof of \reflem{lm:lclp} can be found in \unscite{PuFl_10}.
\end{proof}
\vspace{5pt}
\reflem{lm:lclp} is a generalization of Lemma 3 of \unscite{VoKo_06} to the case of nonbinary codes. One visible difference between the binary case and the present generalization is in the calculation of $V_{i,j,\bar{r}}$ and $C_{j,i,\bar{r}}$. Here in the case of nonbinary codes, the calculation of $V_{i,j,\bar{r}}$ does not exclude the $j$-th entry from $\bs{a} \in \mathcal{A}_i$ and $\bs{\hat{u}}_{i,j}$; similarly, the calculation of $C_{j,i,\bar{r}}$ does not exclude the $i$-th entry from $\bs{b} \in \mathcal{B}_j$ and $\bs{\hat{v}}_{j,i}$. Note that this is not inconsistent since $\hat{u}^{({r})}_{i,j}$ is never used to update itself. 
Here the calculation of $V_{i,j,\bar{r}}$ and $C_{j,i,\bar{r}}$ requires $\bar{r} \in \Re \setminus \{0, r\}$ and hence $\bs{\xi}(\bar{r})$ is always multiplied with the corresponding $\hat{u}^{(\bar{r})}_{i,j}$. This ensures that $\hat{u}^{({r})}_{i,j}$ is not used for calculating $\hat{u}^{*(r)}_{i,j}$.

As mentioned in \unscite{VoKo_06}, the update equation given in Lemma 3 of \unscite{VoKo_06}
can be efficiently computed with the help of the variable and check node calculations of the (binary) SP algorithm. Due to this, the complexity of computing $(C_{j,i,\bar{r}} - C_{j,i,r})$ is $O(d)$ for binary codes.
On the other hand, in the case of nonbinary codes the mapping $\bs{\Xi}$ used in {\bf NBLPD} transforms the nonbinary linear codes $\mathcal{A}_i$ (repetition code) and $\mathcal{B}_j$ (SPC code) into nonlinear binary codes $\mathcal{A}^{\mathrm{NL}}_i = \{\bs{\Xi}(\bs{a}) : \bs{a} \in \mathcal{A}_i\}$ and $\mathcal{B}^{\mathrm{NL}}_j = \{\bs{\Xi}(\bs{b}) : \bs{b} \in \mathcal{B}_j\}$, respectively.
Here, the computation of $(V_{i,j,\bar{r}} - V_{i,j,r})$ and $(C_{j,i,\bar{r}} - C_{j,i,r})$ is related to the SP decoding of nonlinear binary codes 
$\mathcal{A}^{\mathrm{NL}}_i$ and $\mathcal{B}^{\mathrm{NL}}_j$. If $\mathcal{A}_i$ and $\mathcal{B}_j$ have equal lengths then they are duals of each other; however, the relationship between $\mathcal{A}^{\mathrm{NL}}_i$ and $\mathcal{B}^{\mathrm{NL}}_j$ 
is not so simple.

One option to compute $(C_{j,i,\bar{r}} - C_{j,i,r})$ 
is by going through all possible codewords of the SPC code $\mathcal{B}_j$ exhaustively.
In this case the complexity of computing $(C_{j,i,\bar{r}} - C_{j,i,r})$ is $O(d q^{(d-1)})$. Another possibility is to use the trellis of the nonbinary SPC code to calculate these values. In Section \ref{sec:bcjr} we prove that the computation of $C_{j,i,\bar{r}}$ and $C_{j,i,r}$ can be carried out with complexity linear in the check node degree by using a trellis-based variant of the SP algorithm. 

Before we come to that section, we formulate the complete decoding algorithm
which uses the update equation given in \reflem{lm:lclp}. We select an edge
$(i,j) \in \mathcal{E}$, a group element $r \in \Re^{-}$, and calculate
$\hat{u}^{*(r)}_{i,j}$ from \reflem{lm:lclp}. Then $\hat{\phi}_i, \;
\hat{\theta}_j$, and the objective function are updated accordingly. One
\textit{iteration} is completed when all variables associated with all edges
$(i,j) \in \mathcal{E}$ and ring elements $r \in \Re^{-}$ are updated
cyclically. This is a coordinate-ascent type algorithm and its convergence may
be proved in the same manner as in Lemma 4 of \unscite{VoKo_06}.
\begin{lemma} \label{le:convergence} Assume that $d_j \geq 3$, $\forall j \in
  \mathcal{J}$, for a given parity-check matrix $\mathcal{H}$ of the code
  $\mathcal{C}$. If we update the variables associated with all edges $(i,j)
  \in \mathcal{E}$ and ring elements $r \in \Re^{-}$ cyclically with the
  update equation given in \reflem{lm:lclp}, then the objective function of
  {\bf SDNBLPD} converges to its maximum.
\end{lemma}
\begin{proof}
The proof is essentially the same as that of Lemma 4 of \unscite{VoKo_06}.
\end{proof}
\medskip

The algorithm terminates after a fixed number of iterations or when it finds a codeword. Knowing the solution of {\bf SDNBLPD} does not give an estimate of the codeword directly. However, an estimate of the $i$-th symbol $\bs{c}^*_i$ can be obtained from the vector $\bs{\hat{u}}_{i}$. For this we define
\begin{equation*}
\hat{{x}}^{(r)}_i \triangleq \left\{ \begin{array}{cc}
\lambda_{i}^{(r)} - \sum_{j \in \mathcal{J}_i} \hat{u}_{i,j}^{(r)} & \textrm{ if } r \in \Re^{-} \\
0 & \textrm{ if } r = 0 \; . \end{array}\right. \;
\end{equation*}
%
%
Let $\mathcal{M}_i = \argmin_{r \in \Re} \{ \hat{{x}}^{(r)}_i \}$. If $\mathcal{M}_i$ contains a single element $r^*$, then the symbol estimate is obtained as $c^*_i = r^*$; otherwise, we mark $c^*_i$ as \emph{erased}. 

%
%

Due to the soft-minimum operator, the function $\hat{h}
\left(\hat{u}^{(r)}_{i,j}\right)$ in \refeqwo{eq:local_fun}
is differentiable everywhere and this fact is used in \reflem{lm:lclp} to obtain the update equations.
However, for practical implementations we are interested in $\kappa \to \infty$. As mentioned earlier, in the limit $\kappa \to \infty$, the soft-minimum operator becomes the minimum operator, which requires less computation. The following lemma considers $\kappa \to \infty$.
\begin{lemma} \label{le:lemma_k_infty}
In the limit $\kappa \to \infty$, the function $\hat{h}(\hat{u}^{(r)}_{i,j})$ is maximized by any value $\hat{{u}}_{i,j}^{(r)}$ that lies in the closed interval between
\begin{eqnarray*}
(V^{\infty}_{i,j,\bar{r}} -V^{\infty}_{i,j,r})  &\text{ and }& -(C^{\infty}_{j,i,\bar{r}} -C^{\infty}_{j,i,r})
\end{eqnarray*}
where
\begin{eqnarray*}
V^{\infty}_{i,j,\bar{r}} \triangleq - \min_{\underset{\scriptstyle a_j \ne r}{\bs{a} \in \mathcal{A}_i}} \left\langle -\bs{\hat{u}}_i, \bs{\Xi}({\bs{a}}) \right\rangle && C^{\infty}_{j,i,\bar{r}} \triangleq - \min_{\underset{\scriptstyle b_i \ne r}{\bs{b} \in \mathcal{B}_j}} \left\langle -\bs{\hat{v}}_j, \bs{\Xi}(\bs{b}) \right\rangle, \\
V^{\infty}_{i,j,r} \triangleq - \min_{\underset{\scriptstyle a_j = r}{\bs{a} \in \mathcal{A}_i}} \left\langle -\tilde{\bs{u}}_i, \bs{\Xi}(\tilde{\bs{a}}) \right\rangle && C^{\infty}_{j,i,r} \triangleq - \min_{\underset{\scriptstyle b_i = r}{\bs{b} \in \mathcal{B}_j}} \langle -\tilde{\bs{v}}_j, \bs{\Xi}(\tilde{\bs{b}}) \rangle .
\end{eqnarray*}
\end{lemma}
\begin{proof}
The proof of the lemma is a generalization of Lemma 5 of \unscite{VoKo_06}.
\end{proof}
\medskip
\begin{conjecture}
  It is possible to update the variables associated with the edges $(i,j) \in
  \mathcal{E}$ and the ring elements $r \in \Re^{-}$ cyclically, where
  $\hat{{u}}^{*(r)}_{i,j}$ is calculated according to
  \reflem{le:lemma_k_infty}. The authors believe that with a suitable update
  schedule such an algorithm cannot get stuck in a suboptimal point, and that
  the objective function should converge towards the optimal solution of {\bf
    DNBLPD}.
However, in this case it is difficult to prove the convergence of the algorithm. This is because for $\kappa \to \infty$ the objective function is not everywhere differentiable and it is not possible to use the same argument as in \reflem{le:convergence}. This problem is also discussed for the binary case in Conjecture~6 and Section~E of \unscite{VoKo_06}. \endproof
\end{conjecture} 
\vspace{8pt}

After the algorithm terminates, the decision rule described above can be used
to obtain each symbol estimate $c^*_i, i \in \mathcal{I}$.
The nonbinary basic LCLP decoding algorithm of \reflem{lm:lclp} updates a
single variable associated with an edge $(i, j) \in \mathcal{E}$ and a ring
element $r \in \Re^{-}$ at a time. However, we observed from our simulation
work that updating all variables related to an edge $(i, j) \in \mathcal{E}$
simultaneously and processing each edge $(i, j) \in \mathcal{E}$ one at a
time, does not effect the convergence or the error-correcting performance of
the nonbinary basic LCLP decoding algorithm. It is also possible to solve
\textbf{NBLPD} by varying all the edge variables related to a VN $i \in
\mathcal{I}$ or a CN $j \in \mathcal{J}$ simultaneously. Such a variant was
proposed for the basic LCLP algorithm in \unscite{VoSh_09}. We extended the
work of \unscite{VoSh_09} to nonbinary codes for the case in which all the
edge variables related to a VN $i \in \mathcal{I}$ are updated
simultaneously. Details about this case can be found in \unscite{Pu_12}.
For the other case in which all edge variables related to a CN $j \in \mathcal{J}$ are updated simultaneously, we remark that the approach of \unscite{VoSh_09} cannot be used with the nonbinary basic LCLP decoding algorithm. Again, the interested reader is referred to \unscite{Pu_12} for details. 
%
%
%
%
\section{Modified BCJR Algorithm for \\ Check Node Calculation} \label{sec:bcjr}
In this section we propose a modified BCJR algorithm which allows for efficient implementation of the nonbinary basic LCLP decoding algorithm. We observe that the equations for $C_{j,i,\bar{r}}$ and $C_{j,i,r}$ defined in Lemma~\ref{lm:lclp} can be rewritten as follows:
\begin{align}
& \sym \Big(C_{j,i,r}\Big) = \sum_{\underset{\scriptstyle b_i = r}{\bs{b} \in \mathcal{B}_j}} \sym \Big(\langle \tilde{\bs{v}}_j, \bs{\Xi}(\tilde{\bs{b}}) \rangle\Big) \; , \label{eq:c_alpha} \\
&\sym \Big(C_{j,i,\bar{r}}\Big) = \sum_{\underset{\scriptstyle b_i \ne r}{\bs{b} \in \mathcal{B}_j}} \sym \Big(\left\langle \bs{\hat{v}}_j, \bs{\Xi}(\bs{b}) \right\rangle \Big) \; . \label{eq:c_alpha_bar}
\end{align}

It may be observed from the above equations that the calculation of $C_{j,i,r} \; \text{and} \; C_{j,i,\bar{r}}$ is in the form of the marginalization of a product of functions. Hence it is possible to compute $C_{j,i,r} \; \text{and} \; C_{j,i,\bar{r}}$ with the help of a trellis-based variant of the SP algorithm (i.e., a BCJR-type algorithm). One possibility is to use the trellis of the binary nonlinear code $\mathcal{C}^{\mathrm{NL}}_j = \{\bs{\Xi}(\bs{b}) : \bs{b} \in \mathcal{B}_j\}$. However, due to the nonlinear nature of this binary code, the state complexity at the center of its trellis would be exponential in $d_j$. Here state merging is also not possible. Hence there is no complexity advantage when we use the trellis of the binary nonlinear code $\mathcal{C}^{\mathrm{NL}}_j$.

However, if the trellis for the nonbinary SPC code $\mathcal{B}_j$ is used, then the state complexity at each trellis step is $\mathcal{O}(q)$ and is independent of $d_j$. The branch complexity of this trellis is $\mathcal{O}(q^2)$. In the following, we prove that the marginals $C_{j,i,\bar{r}}$ and $C_{j,i,r}$ can be efficiently calculated with some modifications to the BCJR algorithm which uses the trellis of the nonbinary code $\mathcal{B}_j$. 

For ease of exposition, we will assume here that $\cI_j = \{ 1,\ldots, d_j \}$, and let $h_i = \mathcal{H}_{j,i}$ for $i \in \mathcal{I}_j$. We then define the following for the trellis of the SPC code $\mathcal{B}_j$:
\begin{enumerate}
\item The set of all states at time $i'$ is given by $\mathcal{S}_{i'} =
  \Re$, $i' \in \{ 0, \ldots, d_j \}$.
\item There is a branch joining $s \in \mathcal{S}_{i'-1}$ and $s' \in
  \mathcal{S}_{i'}$ for every symbol $b_{i'}$ satisfying $s'-s = h_{i'}
  b_{i'}$ (if no such symbols $b_{i'}$ exist, there is no such trellis
  branch). For such a symbol $b_{i'}$, the ``branch metric'' is given by
  $g(b_{i'}) = \sym \left( \langle\bs{\hat{v}}_{j,i'}\;,\;
    \bs{\xi}(b_{i'})\rangle\right)$.
\item We define
  $\sigma(i'_1, i'_2) = \sum_{i' = i'_1}^{i' = i'_2} \;\; h_{i'}
  b_{i'}$ for $\bs{b} \in \mathcal{B}_j$. In the trellis for the SPC code,
  each state $s \in \mathcal{S}_{i'}$ represents the ``partial syndrome''
  $\sigma(1, i')$.
\item The state metric for forward recursion is
\begin{align}
&\mu_{i'}(s) = \sum_{\underset{\scriptstyle \sigma(1, i') =
    s}{(b_1,\ldots,b_{i'})}} \; \prod_{i'' = 1}^{i'} g(b_{i''}), \quad s \in
\mathcal{S}_{i'}, \ i' \in \mathcal{I}_j \label{eq:mu} 
\end{align}
with $\mu_{0}({0}) = 1$, $\mu_{0}({r}) = 0$, $\forall r \in \Re^{-}$.
Similarly, the state metric for backward recursion is  
\begin{align}
&\nu_{i'}({s}) = \sum_{\underset{\scriptstyle \sigma(i'+1,d_j) =
    -s}{(b_{i'+1},\ldots,b_{d_j})}} \; \prod_{i''=i'+1}^{d_j} g(b_{i''}), \quad s
\in \mathcal{S}_{i'}, \ i' \in \mathcal{I}_j \label{eq:nu}
\end{align}
with $\nu_{d_j}({0}) = 1$, $\nu_{d_j}({r}) = 0$, $\forall r \in \Re^{-}$.
%
%
\end{enumerate}
%

%
%
%
%
\begin{lemma} \label{lm:bcjr_compute}
$C_{j,i,r}$ and $C_{j,i,\bar{r}}$ can be efficiently computed on the trellis of the nonbinary code $\mathcal{B}_j$ as follows,
\begin{align}
& \sym \Big(C_{j,i,\bar{r}}\Big) = \sum_{s \in \mathcal{S}_{i}} \sum_{b_{i}
  \in R \backslash\{ r \}} \mu_{i-1}({s - h_i b_i}) \cdot \nu_{i}({s}) \cdot g(b_i) \label{eq:alpha_bar} \; ,\\
& \sym \Big(C_{j,i,r}\Big) = \sum_{s \in \mathcal{S}_{i}} \mu_{i-1}({s - h_i r}) \cdot \nu_{i}({s}) \; ,\label{eq:alpha}
\end{align}
where state metrics $\mu_{i'}$ and $\nu_{i'}$ are calculated recursively from previous state metrics via
\begin{align*}
&\mu_{i'}(s) = \sum_{{b_{i'} \in \Re}} \mu_{i'-1}({s - h_{i'} b_{i'}}) \cdot g(b_{i'}) \; , \\ &\nu_{i'}({s}) = \sum_{\underset{}{b_{i'+1} \in \Re}} \;\; \nu_{i'+1}({s+h_{i'+1} b_{i'+1}}) \cdot g(b_{i'+1}) \; .
\end{align*}
\end{lemma}
\begin{proof}
The proof of \reflem{lm:bcjr_compute} can be found in \unscite{PuFl_11} for the case where all of the coefficients $\cH_{j,i}$ (for $i \in \cI_j$) equal the ring's multiplicative identity $1$; extension of the proof to handle arbitrary coefficients is straightforward.
\end{proof}
\vspace{5pt}

Here the CN calculations are carried out in two phases: in the first phase,
the forward and backward state metrics are calculated and stored; in the
second phase the marginals $C_{j,i,r}$ and $C_{j,i,\bar{r}}$ are computed
according to \reflem{lm:bcjr_compute}, where the state metrics computed in
first phase are utilized. It may be observed that the aforementioned algorithm
is essentially the same as the BCJR algorithm except for the second phase
where the marginals are calculated.
Note that in general the trellis may contain parallel branches, since some of the entries of the parity-check matrix may be non-invertible elements of the ring.
\subsection{Alternative State Metric for Faster Calculation of $C_{j,i,\bar{r}}$}
The forward state metric $\mu$ as defined in \refeqwo{eq:mu} needs to be
computed for the calculation of $C_{j,i,r}$ and can be reused for the
calculation of $C_{j,i,\bar{r}}$. In \refeqwo{eq:alpha_bar} the algorithm
needs to go through all branches $(s, s') \in \mathcal{S}_{i-1} \times
\mathcal{S}_{i}$, $s' - s \ne h_i r$ for the calculation of
$C_{j,i,\bar{r}}$. If the proposed algorithm is implemented in hardware or on
multicore architectures, then the computation time for $C_{j,i,\bar{r}}$ can
be reduced by parallelizing its calculation. One possibility to parallelize
the calculation of $C_{j,i,\bar{r}}$ is to define a new forward state metric
$\bar{\mu}$ which can be computed in parallel with $\mu$ in the first phase
and reduces the calculations required during the second phase of the
algorithm. For this we define an alternative forward state metric as follows,
\begin{align}
\bar{\mu}_{i'}({s,r}) = \!\!\!\!\!\! \sum_{\underset{\scriptstyle \sigma(1, {i'}) = s, \ b_{i'} \ne r}{(b_1,\ldots,b_{i'})}} \prod_{i'' = 1}^{i'} g(b_{i''}), \;\; s \in \mathcal{S}_{i'}, i' \in \mathcal{I}_j, r \in \Re^{-} \label{eq:new_mu}
\end{align}
with $\bar{\mu}_{0}({s,r}) = 0, \; \forall s \in \mathcal{S}_0, \; \forall r \in \Re^{-}.$ 
It should be noted that due to the condition $b_{i'} \ne r$,
$\bar{\mu}_{i'}({s,{r}})$ cannot be calculated recursively from
$\bar{\mu}_{i'-1}$; instead it is calculated together with $\mu_{i'}$ from
$\mu_{i'-1}$ as follows,
\begin{align*}
&\bar{\mu}_{i'}({s,{r}}) = \sum_{\underset{}{b_{i'} \in \Re\setminus\{r\}}} \mu_{i'-1}({s - h_{i'} b_{i'}}) \cdot g(b_{i'}).
\end{align*}
With the help of the alternative forward state metric given in \eqref{eq:new_mu}, the expression \refeqwo{eq:alpha_bar} of \reflem{lm:bcjr_compute} can be rewritten as
\begin{align}
\sym\Big(C_{j,i,\bar{r}}\Big) = \sum_{s \in \mathcal{S}_{i}} \bar{\mu}_{i}({s,{r}}) \cdot \nu_{i}(s) . \label{eq:new_alpha_bar}
\end{align}

The forward state metric $\bar{\mu}_{i'}({s,r})$ requires the calculation and storage of an additional $q-1$ values for each state $s \in \mathcal{S}_{i'}$ during the first phase. Hence the storage requirement for the calculation of $C_{j,i,\bar{r}}$ with \refeqwo{eq:new_alpha_bar} increases by a factor of $q$.
However, all additional state metric values can be calculated in parallel with $\mu$ which does not effect the run time of the first phase of the algorithm. Also, the second phase of the algorithm needs to go through only $q$ states instead of $q(q-1)$ branches, hence the overall run time for computing $C_{j,i,\bar{r}}$ is reduced with the state metric $\bar{\mu}$.
\subsection{Calculation of Marginals with $\kappa \to \infty$}
In \reflem{lm:bcjr_compute}, $\kappa$ is assumed to be finite. However, for many practical applications we are interested in $\kappa \to \infty$. According to \reflem{le:lemma_k_infty}, for $\kappa \to \infty$ we need to calculate $(C_{j,i,r}^{\infty} - C_{j,i,\bar{r}}^{\infty})$ to update the corresponding variables.
The marginals $C_{j,i,r}^{\infty}$ and $C_{j,i,\bar{r}}^{\infty}$ are here obtained as the limit $\kappa \to \infty$ of \refeqwo{eq:c_alpha} and \refeqwo{eq:c_alpha_bar}, respectively, i.e.,
\begin{equation}
C_{j,i,r}^{\infty} \triangleq - \min_{\underset{\scriptstyle b_i = r}{\bs{b} \in \mathcal{B}_j}} \langle -\tilde{\bs{v}}_j, \bs{\Xi}(\tilde{\bs{b}}) \rangle, \; 
C_{j,i,\bar{r}}^{\infty} \triangleq - \min_{\underset{\scriptstyle b_i \ne r}{\bs{b} \in \mathcal{B}_j}} \left\langle -\bs{\hat{v}}_j, \bs{\Xi}(\bs{b}) \right\rangle .\label{eq:new_c}
\end{equation}
Thus $C_{j,i,r}^{\infty}$ and $C_{j,i,\bar{r}}^{\infty}$ can be obtained by replacing all ``product'' operations with ``sum'' operations and similarly by replacing all ``sum'' operations with ``min'' operations in \refeqwo{eq:c_alpha} and \refeqwo{eq:c_alpha_bar} (marginals with finite $\kappa$). In \refeqwo{eq:c_alpha} and \refeqwo{eq:c_alpha_bar} the marginalization is performed in the sum-product semiring. However, for $\kappa \to \infty$ the marginalization is performed in the min-sum semiring and hence the marginals of \refeqwo{eq:new_c} can be computed with a trellis-based variant of the MS algorithm. If we replace all ``product'' operations with ``sum'' operations and all ``sum'' operations with ``min'' operations in \refeqwo{eq:mu}, \refeqwo{eq:nu}, \refeqwo{eq:alpha_bar}, \refeqwo{eq:alpha}, \refeqwo{eq:new_mu} and \refeqwo{eq:new_alpha_bar}, and then redefine the branch metric as $g(b_{i'}) = \langle\bs{\hat{v}}_{j,i'}, \bs{\xi}(b_{i'})\rangle$, then the resulting equations 
can be used on the trellis of the nonbinary SPC code $\mathcal{B}_j$ to compute the marginals of \refeqwo{eq:new_c}. This trellis-based variant of the MS algorithm is related to the Viterbi algorithm.
%
%
%
%
%
%
%
\section{Nonbinary Subgradient Low-Complexity \\ LP Decoding Algorithm}
In \unscite{VoKo_06} the authors proposed the subgradient LCLP decoding algorithm for binary LDPC codes. The objective function of the dual LP (denoted \textbf{DLPD2} in \unscite{VoKo_06}) can be expressed as a sum of several component functions. Based on this observation, the authors proposed the use of incremental subgradient methods \unscite{Ne_02} for the maximization of the dual objective function in \textbf{DLPD2}. 

The main idea behind incremental subgradient methods is to process each component function separately where variables related to the selected component function are updated immediately. An iteration of the incremental subgradient method can be seen as a sequence within which each component function is processed exactly once \unscite{Ne_02}.

Similar to the dual LP {\bf DLPD} of \unscite{VoKo_06}, the objective function (which is concave but not everywhere differentiable) of \textbf{DNBLPD} can also be expressed as the sum of component functions. Hence it is also possible to use incremental subgradient methods to find the solution of \textbf{DNBLPD}. As in the previous section, we assume $\bs{\hat{u}}_{i,j} = -\bs{\hat{v}}_{j,i}$ for all $(i, j) \in \cE$.

To develop the nonbinary subgradient LCLP decoding algorithm, we consider the component function given by the term in the objective function related to CN $j \in \mathcal{J}$, i.e.,
\begin{align} 
m_j \left( \bs{\hat{v}}_j \right) = \min_{\bs{b} \in \mathcal{B}_j} \left\langle-\bs{\hat{v}}_j, \bs{\Xi}(\bs{b})\right\rangle \; . \label{eq:sub_local_fun}
\end{align}
We provide the definition of the subgradient for this part of the objective function in the following lemma.

\begin{lemma} \label{lm:sub}
For the term in the objective function related to the CN  $j \in \mathcal{J}$ given in \refeqwo{eq:sub_local_fun},
a subgradient is given by
\begin{align*}
\bs{s}_j(\bs{\hat{v}}_j) &= -\bs{\Xi} \left(\argmin_{\bs{b} \in \mathcal{B}_j}  \langle -\bs{\hat{v}}_j, \; \bs{\Xi}(\bs{b}) \rangle \right) .
\end{align*}
\end{lemma}
\begin{proof}
For $\bs{s}_j(\bs{\hat{v}}_j)$ to be a subgradient of $m_j(\bs{\hat{v}}_j)$, the following inequality must hold \unscite{Ne_02}
\begin{align}
m_j(\bs{v}_j') \le m_j(\bs{\hat{v}}_j) + \langle \bs{s}_j(\bs{\hat{v}}_j), \;  \bs{v}_j' - \bs{\hat{v}}_j \rangle \label{eq:to_prove}
\end{align}
for all $\bs{v}_j' \in \mathbb{R}^{(q-1)|\mathcal{I}_j|}$. We define
\begin{equation}
\bs{b}' \triangleq \argmin_{\bs{b} \in \mathcal{B}_j}  \langle -\bs{\hat{v}}_j, \; \bs{\Xi}(\bs{b}) \rangle.
\label{eq:bprime}
\end{equation}
With this, we obtain
\begin{align*}
&m_j(\bs{\hat{v}}_j) + \langle \bs{s}_j(\bs{\hat{v}}_j), \;  \bs{v}_j' - \bs{\hat{v}}_j \rangle \\
&= \min_{\bs{b} \in \mathcal{B}_j}  \langle -\bs{\hat{v}}_j, \; \bs{\Xi}(\bs{b}) \rangle + \langle \bs{s}_j(\bs{\hat{v}}_j), \;  \bs{v}_j' \rangle - \langle \bs{s}_j(\bs{\hat{v}}), \; \bs{\hat{v}}_j \rangle \\
&= \langle -\bs{\hat{v}}_j, \; \bs{\Xi}(\bs{b}') \rangle + \langle -\bs{\Xi}(\bs{b}'), \;  \bs{v}_j' \rangle - \langle -\bs{\Xi}(\bs{b}'), \; \bs{\hat{v}}_j \rangle \\
&= \langle -\bs{v}_j', \; \bs{\Xi}(\bs{b}') \rangle \\
&\ge \min_{\bs{b} \in \mathcal{B}_j}  \langle -\bs{v}_j', \; \bs{\Xi}(\bs{b}) \rangle \\
&= m_j(\bs{v}_j') \; ,
\end{align*}
thereby proving \eqref{eq:to_prove} and the fact that $\bs{s}_j(\bs{\hat{v}}_j)$ is a subgradient of $m_j(\bs{\hat{v}}_j)$.

Note that if more than one vector $\bs{b} \in \mathcal{B}_j$ achieves the minimum in \eqref{eq:bprime}, a subgradient is given by the negative of an arbitrary linear combination of the corresponding vectors $\bs{\Xi}(\bs{b})$.
\end{proof}
\vspace{5pt}

The subgradient of the above lemma is denoted by {\bf S}${}_{\bs \Xi}$.
It can be observed that the subgradient {\bf S}${}_{\bs \Xi}$, which is a generalization of the subgradient given in \unscite{VoKo_06} for binary codes, can be efficiently obtained with the help of the Viterbi algorithm on the trellis of the nonbinary SPC code $\mathcal{B}_j$. Once the subgradient is obtained, the dual variable $\bs{\hat{v}}_j$ can be updated as \unscite{Ne_02}
\begin{align}
\bs{\hat{v}}_j \leftarrow \bs{\hat{v}}_j + \vartheta_l \cdot \bs{s}_j(\bs{\hat{v}}_j), \label{eq:sub_update}
\end{align}
where $\vartheta_l \in \mathbb{R}_{>0}$ is the step size at iteration $l$. The dual variable  $\bs{\hat{u}}_i$ related to VN $i \in \mathcal{I}$ can be updated in an analogous manner. 
The subgradient for the VN update can be computed with some modifications to the VN calculations used in the nonbinary SP algorithm.
One iteration of the algorithm is completed when all check-node-related updates of dual variables $\bs{\hat{v}}_j, \; j \in \mathcal{J}$, and then, all variable-node-related updates of dual variables $\bs{\hat{u}}_i, \; i \in \mathcal{I}$, have been (sequentially) performed.
The convergence of this algorithm is guaranteed for a suitably chosen step size sequence $\{\vartheta_l\}_{l \ge 1}$ \unscite{Ne_02}. The decision rule to obtain the estimate of the symbols from the dual variables $\bs{\hat{u}}_i, \; i \in \mathcal{I}$, is the same as the one given in \refsec{sec:lclp}.

\begin{figure*}
\begin{minipage}[b]{0.48\linewidth}
\centering
\includegraphics[width=0.99\columnwidth, keepaspectratio]{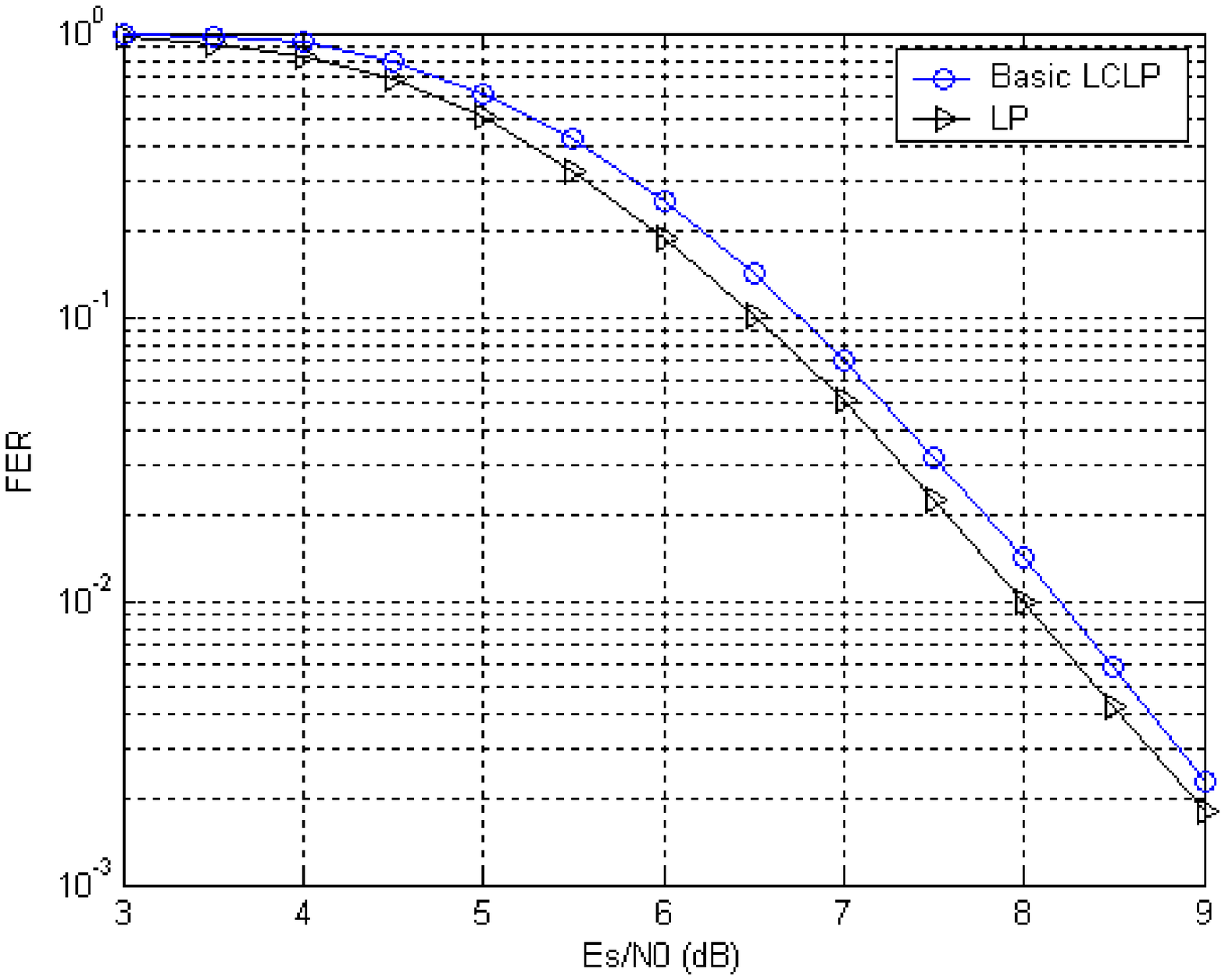}
\vspace{-10pt}
\caption{Frame error rate for the $(80, 48)$ LDPC code over $\mathbb{Z}_4$ under QPSK modulation.} 
\label{fig:FER_80_Q4}
\vspace{-10pt}
\end{minipage}
\hspace{0.4cm}
\begin{minipage}[b]{0.48\linewidth}
\includegraphics[width=1.0\columnwidth, keepaspectratio]{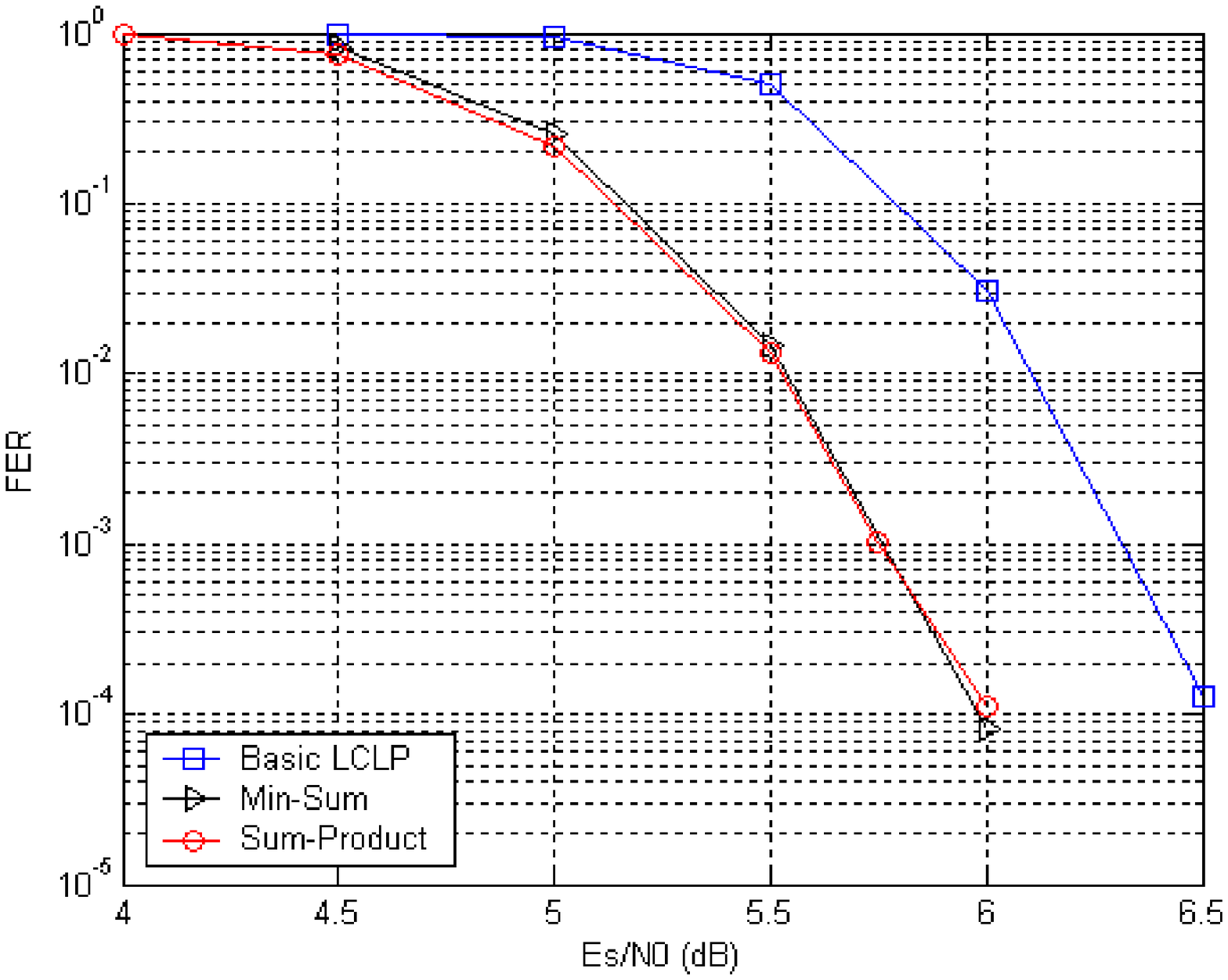}
\centering
\caption{Frame error rate for the $(1055, 424)$ LDPC code over GF($4$) under QPSK modulation.} 
\label{fig:FER_1055_Q4}
\vspace{-10pt}
\end{minipage}
\end{figure*}

The choice of step size sequence $\{\vartheta_l\}_{l \ge 1}$ can also affect the convergence as well as the error-correcting performance of the algorithm. Some possible step size rules (e.g., constant, diminishing, dynamic etc.) are discussed in \unscite{Ne_02}. 
It was determined through extensive simulation works that the following {\it staircase type} step size rule works best for most nonbinary LDPC codes (independent of the code parameters):
\begin{eqnarray*}
\vartheta_l = \left\{ 
\begin{array}{l l}
  \vartheta_{l-1} \times 0.8 & \; \text{if} \; l \text{ is divisible by } 20 \\
  \vartheta_{l-1} & \; \text{otherwise} .\\
\end{array} \right. 
\end{eqnarray*}
The initial value $\vartheta_1$ is also determined by the simulation.

The nonbinary basic LCLP decoding algorithm is an \emph{edge-by-edge}
algorithm, i.e., it processes each edge in the Tanner graph separately. During
the decoding of $(d_{\mathrm{v}}, d_{\mathrm{c}})$-regular LDPC codes with the
nonbinary basic LCLP decoding algorithm, the modified BCJR algorithm of
\reflem{lm:bcjr_compute} is utilized $m \cdot d_{\mathrm{c}}$ times, and VN
calculations are carried out $n \cdot d_v$ times, in a single iteration.\footnote{However, the complexity of the nonbinary basic LCLP
  decoding algorithm can be significantly reduced with the help of a suitable
  update schedule and a suitable reuse of partial results. \label{footnote:remark:complexity:1}} In
contrast to this, the nonbinary subgradient LCLP decoding algorithm works on a
\emph{node-by-node} basis, i.e., it updates all the edges related to a CN or a
VN simultaneously. Hence the nonbinary subgradient algorithm runs the Viterbi
algorithm only $m$ times, and performs the VN calculations $n$ times, in a
single iteration. Also, the Viterbi algorithm is computationally less
expensive than the modified BCJR algorithm used in the nonbinary basic LCLP
decoding algorithm. This reduces the complexity of a single iteration of the
nonbinary subgradient algorithm significantly.  One more advantage of the
nonbinary subgradient LCLP decoding algorithm is the ease of computation of
the dual function value (the contribution of the component function given in
\reflem{lm:sub} towards the global function is computed by the Viterbi
algorithm in the form of the forward state metric).
Similarly, the component function value is also output as a by-product of the VN computations. Hence the global function value can be easily computed during each iteration. The algorithm may be deemed to have converged to the solution of {\bf DNBLPD} when the difference between the global function values computed during successive iterations is close to zero; this criterion may be used to efficiently implement an early stopping mechanism. The global function value computed during each iteration can also be utilized to adapt the step-size dynamically to improve the convergence and/or error-correcting performance of the nonbinary subgradient LCLP decoding algorithm.

The complexity per iteration of the nonbinary SP and MS algorithms is dominated by that of the CN calculation, which is $\mathcal{O}(q^2)$ \unscite{DeFo_07}.
The nonbinary SP and MS algorithms also work on a node-by-node basis, and
consequently the nonbinary SP (resp. MS) algorithm uses the BCJR
(resp. Viterbi) algorithm $m$ times during each iteration. As mentioned
earlier, the nonbinary basic LCLP decoding algorithm uses the modified BCJR
algorithm $m \cdot d_{\mathrm{c}}$ times (however, see
Footnote~\ref{footnote:remark:complexity:1}) and hence its complexity per
iteration is significantly higher than that of the nonbinary SP or MS
algorithm. In contrast, the nonbinary subgradient LCLP decoding algorithm,
which uses the Viterbi algorithm $m$ times during each iteration, has
complexity per iteration similar to that of the nonbinary MS algorithm and
smaller than that of the nonbinary SP algorithm.
\section{Simulation Results} \label{sec:result}
This section presents simulation results for the nonbinary basic and subgradient LCLP decoding algorithms. We use a cyclic edge-update schedule for the nonbinary basic LCLP decoding algorithm. The nonbinary basic LCLP decoding algorithm uses the trellis-based CN calculations described in \refsec{sec:bcjr} and we consider $\kappa \to \infty$ for all simulations. The MS and SP algorithms also use the trellis of the nonbinary SPC code for CN processing. We use the binary $(204, 102)$ MacKay LDPC matrix and the $(155, 64)$, $(755, 334)$, and $(1055, 424)$ group-structured LDPC matrices from \unscite{TaSr_01}, but with nonzero parity-check matrix entries replaced by randomly selected nonzero entries from the finite ring. The $(155, 64)$ and $(1055, 424)$ LDPC codes have parity-check matrix entries from $\mathbb{Z}_4$ and GF($4$), respectively, and the $(204, 102)$ and $(755, 334)$ LDPC codes have parity-check matrix entries from GF($8$). We also use the LDPC code of length $n = 80$ over $\mathbb{Z}_4$ used in \unscite{FlSk_09} which has rate $R(\mathcal{C}) = 0.6$ and constant check-node degree of $5$. 
\begin{figure*}
\begin{minipage}[b]{0.48\linewidth}
\centering
\includegraphics[width=1.0\columnwidth, keepaspectratio]{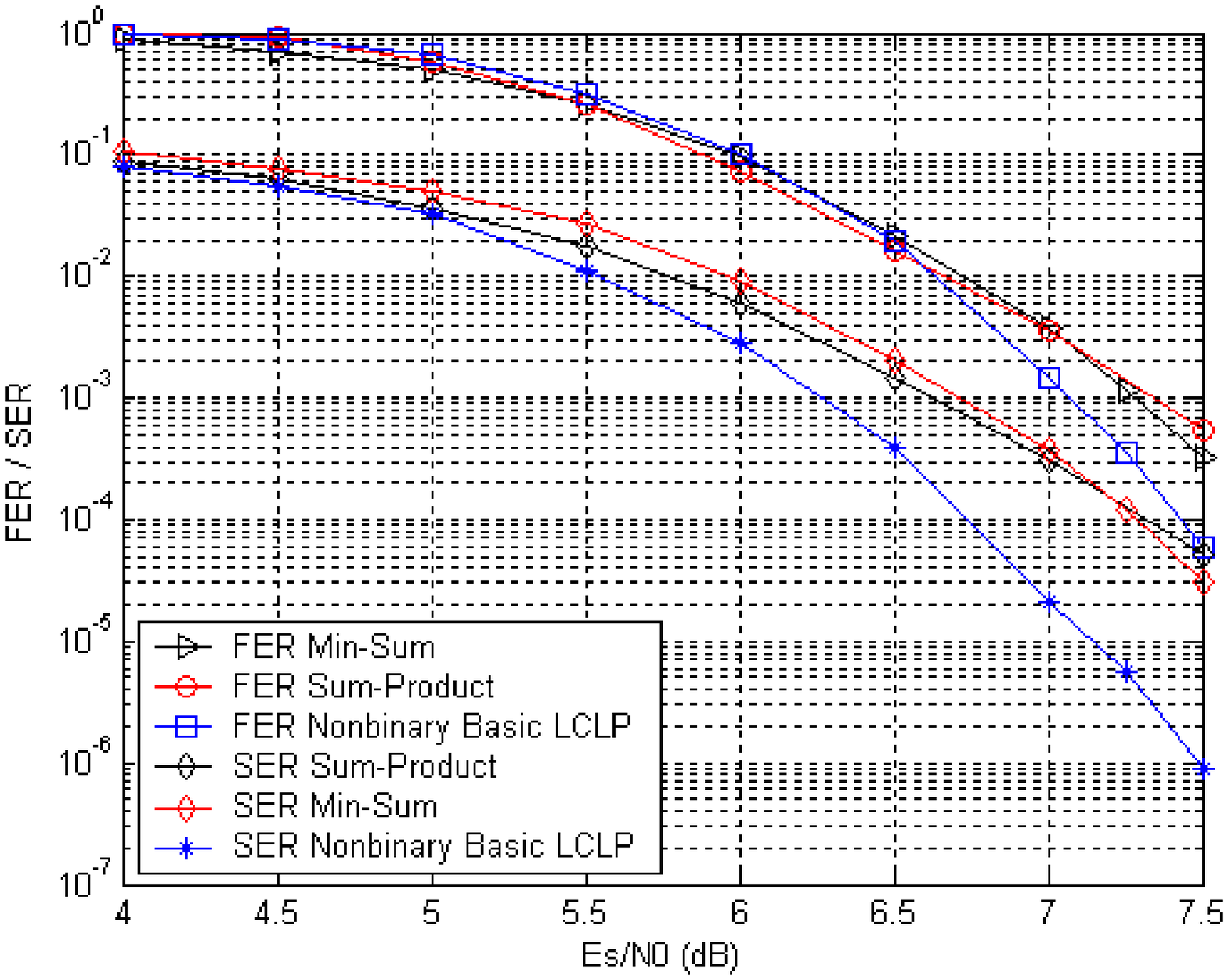}
\caption{Frame/symbol error rate for the $(155, 64)$ quaternary LDPC code under QPSK modulation.} 
\label{fig:FER_155_Q4}
\vspace{-5pt}
\end{minipage}
\hspace{0.4cm}
\begin{minipage}[b]{0.48\linewidth}\centering
\centering
\includegraphics[width=1.0\columnwidth, keepaspectratio]{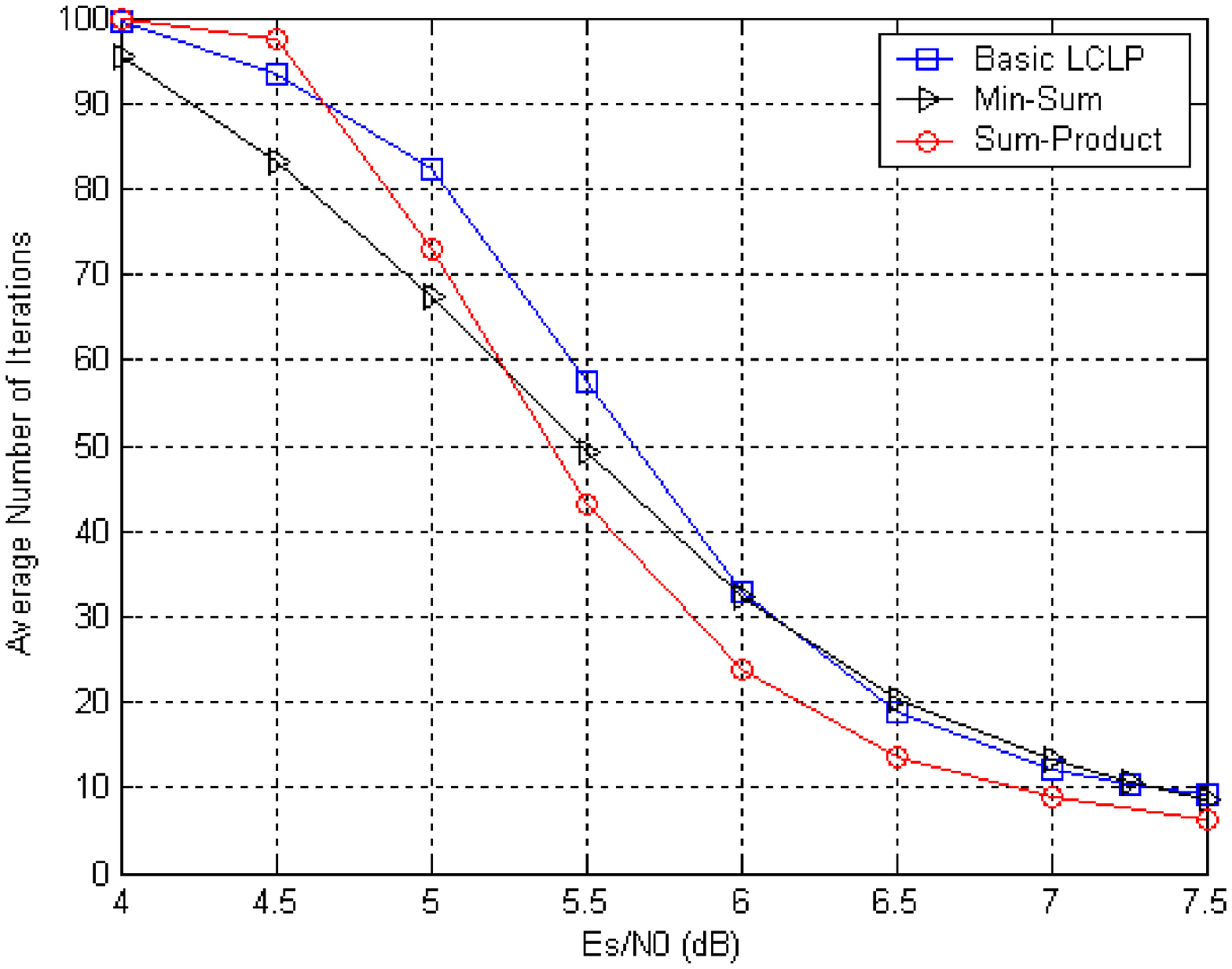}
\caption{Average number of iterations required to converge for the $(155, 64)$ quaternary LDPC code under QPSK modulation.} 
\label{fig:iter_155_Q4}
\vspace{-5pt}
\end{minipage}
\end{figure*}
%
The $(155, 64)$, $(1055, 424)$, and $(755, 334)$ matrices are $(3,5)$-regular group-structured LDPC matrices; hence there are $5$ nonzero entries in each row. For the $(155, 64)$ and $(1055, 424)$ matrices, we set all non-zero entries to $1 \; (=\zeta^0)$ in each row ($\zeta$ is a primitive element of the finite field under consideration). For the $(755, 334)$ LDPC matrix, the first, second, third, fourth, and fifth nonzero entry in each row is set to $1$, $\zeta^2$, $\zeta^4$, $\zeta^6$, and $1$, respectively.
The $(204, 102)$ LDPC matrix is a $(3,6)$-regular matrix and we set the first, second, third, fourth, fifth, and sixth nonzero entry in each row to elements $1$, $\zeta^2$, $\zeta^4$, $\zeta^6$, $\zeta^1$, and $1$, respectively.

Furthermore, we assume transmission over the AWGN channel where for the $(155, 64)$ and $(1055, 424)$ LDPC codes the nonbinary symbols are directly mapped to quaternary phase-shift keying (QPSK) signals and for the $(204, 102)$ and $(755, 334)$ codes, nonbinary symbols are directly mapped to 8-PSK signals. We simulate up to $100$ frame errors per simulation point. Unless otherwise specified, the maximum number of iterations is set to $100$.

\reffig{fig:FER_80_Q4} compares the frame error rate (FER) for nonbinary LP decoding of \unscite{FlSk_09} (solution performed using the Simplex solver) with that of the nonbinary basic LCLP decoding algorithm. As can be observed, the error correcting performance of the nonbinary basic LCLP decoding algorithm is within $0.2$ dB of the LP decoder. Note however that since the nonbinary LCLP decoding algorithm only approximately solves {\bf PNBLPD}, it does not possess the ML certificate property.

The FER curves for the $(1055, 424)$ LDPC code is shown in Figure \reffigwo{fig:FER_1055_Q4} and the FER of nonbinary basic LCLP decoding is within $0.5$dB of that of the SP and MS algorithm.

The error-correcting performance of the $(155, 64)$ LDPC code is shown in Figure \reffigwo{fig:FER_155_Q4} where the FER and symbol error rate (SER) of the nonbinary basic LCLP decoding algorithm is compared with that of the SP and MS algorithm. For this code, the FER of the nonbinary basic LCLP decoding algorithm is similar to that of the SP and MS algorithms for low and moderate SNR levels; however, it is better by around $0.25$dB for higher SNR levels. The SER of the nonbinary basic LCLP decoding algorithm is better than that of the SP and MS algorithms for all tested SNR values. \reffig{fig:iter_155_Q4} shows the average number of iterations required for the nonbinary basic LCLP, SP, and MS algorithms to converge during the decoding of the $(155, 64)$ LDPC code. The nonbinary basic LCLP decoding algorithm requires around $10$\% to $15$\% more iterations than the MS algorithm to converge for lower SNR levels, whereas nonbinary basic LCLP and MS algorithms require a similar number of iterations for moderate to high SNR values 
(i.e., in the waterfall region). 
Hence the nonbinary basic LCLP decoding algorithm outperforms the MS decoding algorithm in terms of the error-correcting performance for the $(155, 64)$ LDPC code. The SP algorithm requires around 30\% less iterations to converge compared to nonbinary basic LCLP decoding algorithm for most SNR values.

The FER curve for the $(204, 102)$ LDPC code is shown in Figure \reffigwo{fig:FER_204_Q8}. In this case the FER of nonbinary basic LCLP decoding is within $1.5$dB and $0.75$dB of that of the SP and MS algorithms, respectively.

\reffig{fig:FER_755_Q8} shows the FER curves for the $(755, 334)$ LDPC code. Unlike the above mentioned results, here the FER performance of the nonbinary basic LCLP decoding algorithm is around $1.2$dB and $1.7$dB worse than that of the SP and MS algorithm, respectively, for low to moderate SNR values. However, for SNR values higher than $9$dB, the SP and MS algorithm shows an error-floor effect and by around $10$dB its FER is the same as that of the nonbinary basic LCLP decoding algorithm. 
After $10.25$dB, the nonbinary basic LCLP decoding algorithm also shows the error floor effect but still has better FER than the MS and SP algorithm. The FER of nonbinary basic LCLP decoding algorithm at $10.25$dB and $10.5$dB was simulated for $60$ frame errors per simulation point for this code.
A similar phenomenon was also observed in \unscite{Bu_09} where the binary LCLP decoding algorithm outperformed the MS algorithm in the error-floor region. It is important to note that the binary $(755, 334)$ LDPC code is constructed with the same algorithm as the other $(3, 5)$ group-structured LDPC codes \unscite{TaSr_01}; however 
its minimum distance is relatively low compared to other binary LDPC codes from the same family, and hence one can expect the binary MS (or SP) algorithm to show an error-floor effect. Our observation of a high error-floor for the $(755, 334)$ LDPC code over GF($8$) could be due to a similar problem with respect to the Lee distance.
\begin{figure*}
\begin{minipage}[b]{0.48\linewidth}\centering
\centering
\includegraphics[width=1.0\columnwidth, keepaspectratio]{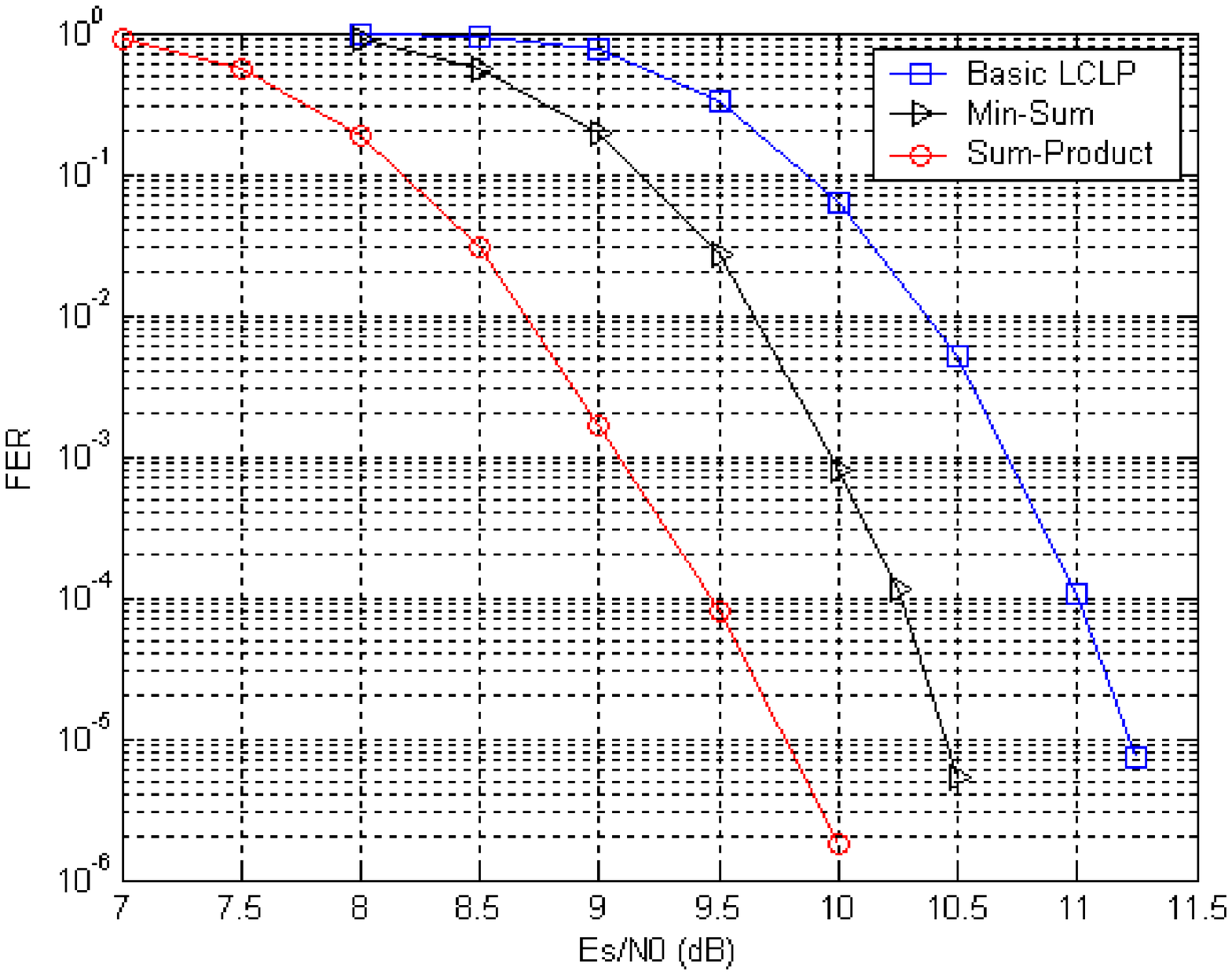}
\caption{Frame error rate for the $(204, 102)$ LDPC code over GF($8$) under 8-PSK modulation.} 
\label{fig:FER_204_Q8}
\vspace{-5pt}
\end{minipage}
\hspace{0.4cm}
\begin{minipage}[b]{0.48\linewidth}\centering
\centering
\includegraphics[width=1.0\columnwidth, keepaspectratio]{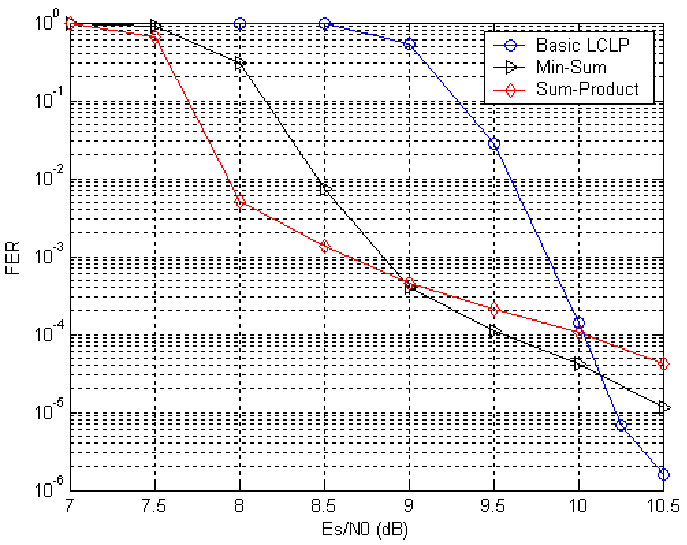}
\caption{Frame error rate for the $(755, 334)$ LDPC code over GF($8$) under 8-PSK modulation.} 
\label{fig:FER_755_Q8}
\vspace{-5pt}
\end{minipage}
\end{figure*}

The FER of the $(155, 64)$ LDPC code over $\mathbb{Z}_4$ for the nonbinary subgradient LCLP decoding algorithm is shown in Figure \reffigwo{fig:FER_155_Q4_sub}. 
The FER of the nonbinary basic LCLP decoding algorithm is also shown here for reference, where the maximum number of iterations is set to $100$. Both the
constant and staircase type step-size rules are used for these simulations. 
Also, \reffig{fig:iter_155_Q4_sub} shows the average number of iterations required for the nonbinary subgradient LCLP decoding algorithm to converge, with different step-size rule combinations (maximum $100$ iterations). The initial value of the step-size at the first iteration for the simulations of \reffig{fig:FER_155_Q4_sub} was optimized through simulation, and for the constant step-size rule it is $0.08$ whereas for the staircase type step-size rule it is $0.15$.
%
%
%
%
The nonbinary subgradient LCLP decoding algorithm with 
staircase type step-size rule
has better FER than the 
constant step-size rule, while requiring a similar average number of iterations to converge.

The FER of the nonbinary subgradient LCLP decoding algorithm with 
staircase type step-size rule
is $0.38$dB away from the FER of the nonbinary basic LCLP decoding algorithm for a maximum of $100$ iterations and is better by $0.06$dB for a maximum of $200$ iterations. However, it requires approximately 3 times as many iterations on average to converge than the nonbinary basic LCLP decoding algorithm. As was already discussed in the previous section, the complexity of a single iteration of the nonbinary subgradient LCLP decoding algorithm is significantly lower than that of the nonbinary basic LCLP decoding algorithm. However, this complexity advantage is somewhat mitigated by the fact that the nonbinary subgradient LCLP decoding algorithm requires a higher number of iterations than the nonbinary basic LCLP decoding algorithm to reach a similar FER for a given SNR value.

For the $(204, 102)$ LDPC code, if 
the same step-size rule and maximum number of iterations is used, then the nonbinary subgradient LCLP decoding algorithm 
requires around $0.75$dB more transmit power than the nonbinary basic LCLP decoding algorithm to reach same FER. 

For the $(1055, 424)$ and the $(755, 334)$ LDPC codes, the FER of the nonbinary subgradient LCLP decoding algorithm which uses 
the staircase type step-size rule (maximum $200$ iterations) is similar to or better than that of the nonbinary basic LCLP decoding algorithm (maximum $100$ iterations). For these simulations, the initial value of the step-size (again optimized through simulation) 
for the $(1055, 424)$ code was $0.20$ and for the $(755, 334)$ code was $0.09$.

\begin{figure*}
\begin{minipage}[b]{0.48\linewidth}
\centering
\includegraphics[width=1.0\columnwidth, keepaspectratio]{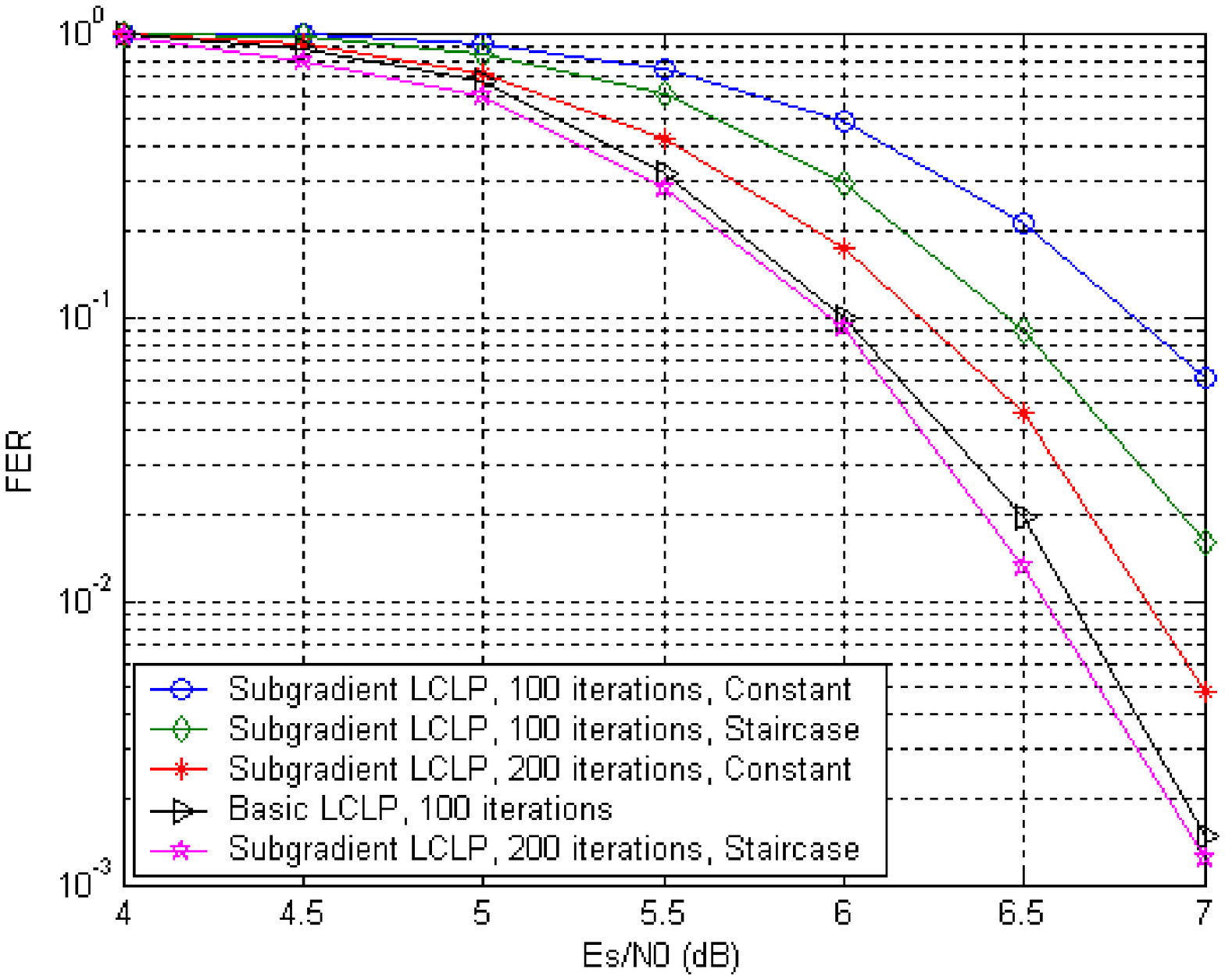}
\caption{Frame error rate for the $(155, 64)$ quaternary LDPC code under QPSK modulation. The FER performance of the nonbinary basic LCLP decoding algorithm is compared with that of the nonbinary subgradient LCLP decoding algorithm (with different step-size rule).} 
\label{fig:FER_155_Q4_sub}
\vspace{-5pt}
\end{minipage}
\hspace{0.4cm}
\begin{minipage}[b]{0.48\linewidth}
\centering
\includegraphics[width=1.0\columnwidth, keepaspectratio]{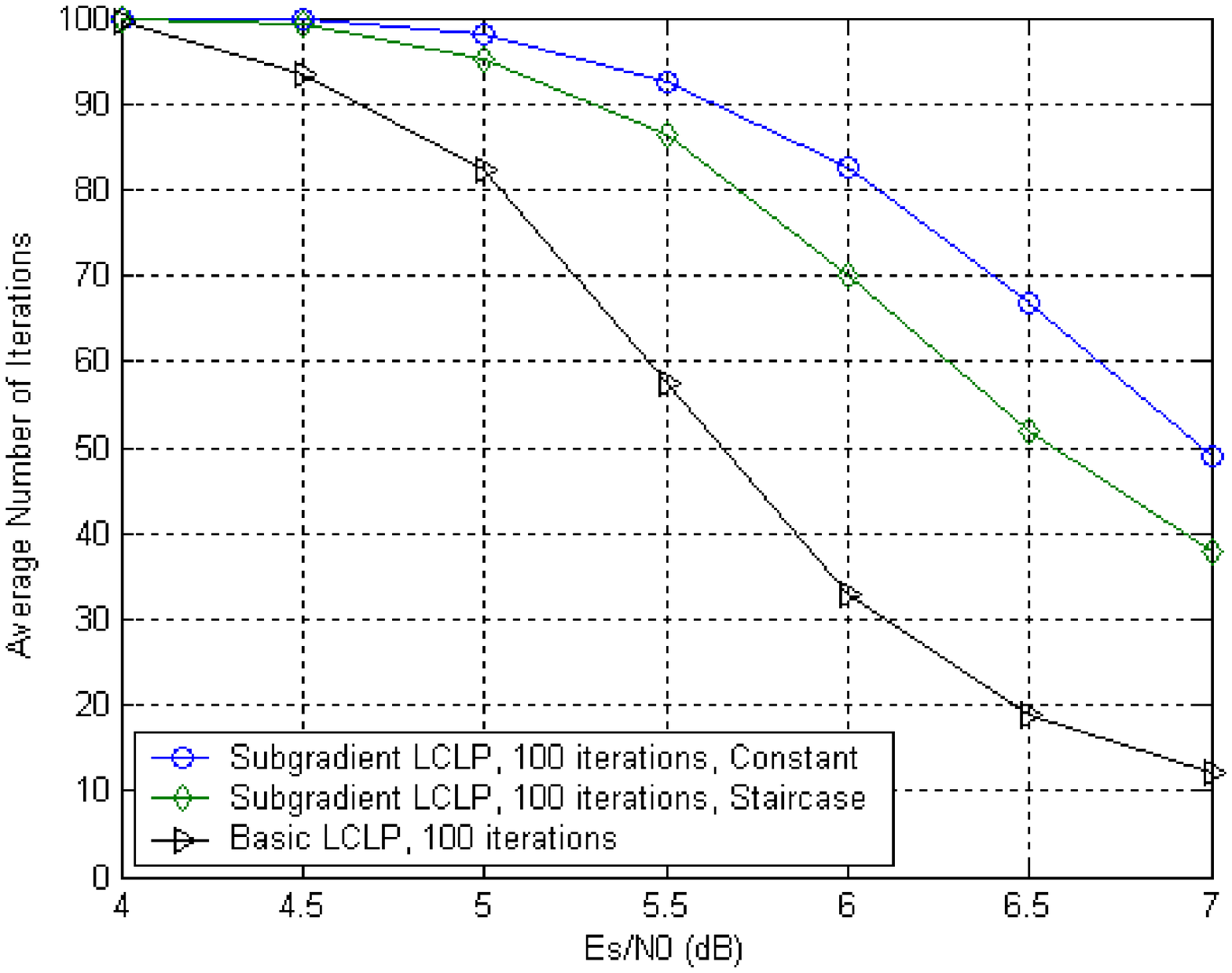}
\caption{Average number of iterations required for the nonbinary basic LCLP and the nonbinary subgradient LCLP decoding algorithm (with different step-size rule) to converge for the $(155, 64)$ quaternary LDPC code under QPSK modulation. } 
\label{fig:iter_155_Q4_sub}
\vspace{-5pt}
\end{minipage}
\end{figure*}

%
\section{Conclusions}
In this paper we generalized the basic LCLP decoding algorithm and the subgradient LCLP decoding algorithm to nonbinary linear codes. %
The complexity per iteration of the nonbinary LCLP decoding algorithms is linear in the code's block length and hence they can also be used for moderate and long block length codes. The complexity of nonbinary basic LCLP decoding algorithm is dominated by the maximum check node degree and the number of elements in the nonbinary alphabet. Furthermore, we proposed a modified BCJR algorithm for efficient check node processing in the nonbinary basic LCLP decoding algorithm. The proposed CN processing algorithm has complexity linear in the check node degree. We also proposed an alternative state metric which can be used to reduce the run time of the CN calculations of the nonbinary basic LCLP decoding algorithm. The error-correcting performance of the nonbinary basic LCLP decoding algorithm is similar to that of the MS algorithm for some classes of LDPC codes. %
%
%
%
%
%
%
%
%
%
%
%
%
%
%
%
%
%

%
%

\section*{ACKNOWLEDGMENTS}
The authors would like to thank Maria Jos\'{e} Canet Subiela for her valuable support in their simulation work.

\mybibitem{Ga_63}{R.~G.~Gallager, {\it Low Density Parity Check Codes}, Monograph, M.I.T. Press, 1963.}

\mybibitem{DaMa_98}{M.~C.~Davey and D.~J.~C.~MacKay, ``Low density parity check codes over GF(q),'' IEEE Communication Letters, vol. 2, no. 6, pp. 165--167, June 1998.}

\mybibitem{Fo_01}{G.~D.~Forney, Jr., ``Codes on graphs: normal realizations,'' IEEE Trans. Inf. Theory, vol. 47, no. 2, pp. 520--548, February 2001.}

\mybibitem{TaSr_01}{R.~M.~Tanner, D.~Sridhara, and T.~Fuja, ``A class of group-structured LDPC codes,'' in Proceedings of Sixth Int. Symp. on Communication Theory and Applications, pp. 365--370, Ambleside, England, July 2001.}

\mybibitem{DiPr_02}{C.~Di, D.~Proietti, I.~E.~Telatar, T. J. Richardson, and R. L. Urbanke, ``Finite-length analysis of low-density parity-check codes on the binary erasure channel,'' IEEE Trans. Inf. Theory, vol. 48, no. 6, pp. 1570--1579, June 2002.}

\mybibitem{Vo_02}{P.~O.~Vontobel, {\it Kalman Filters, Factor Graphs, and Electrical Networks}, Post-Diploma Project, Dept. of Information Technology and Electrical Engineering, ETH Zurich, 2002.}

\mybibitem{Ne_02}{A.~Nedi\'c,  {\it Subgradient Methods for Convex Minimization}. Ph.D. thesis, Dept. of Electrical Engineering and Computer Science, Massachusetts Institute of Technology, Cambridge, MA, 2002.}

\mybibitem{Ri_03}{T.~Richardson,  ``Error floors of LDPC codes,'' in Proceedings of 41st Annual Allerton Conference on Communication, Control, and Computing, pp. 1426--1435, Monticello, IL, October 2003.}

\mybibitem{VoLo_03}{P.~O.~Vontobel and H.-A.~Loeliger, ``On factor graphs and electrical networks,'' Mathematical Systems Theory in Biology, Communication, Computation, and Finance, J. Rosenthal and D.S. Gilliam, eds., IMA Volumes in Math. \& Appl., Springer Verlag, pp. 469--492, 2003.}

\mybibitem{Fe_Th_03}{J.~Feldman, {\it Decoding Error-Correcting Codes via Linear Programming}, Ph.D. Thesis, Dept. of Electrical Engineering and Computer Science, Massachusetts Institute of Technology, Cambridge, MA, June 2003.}

\mybibitem{LO_04}{H.-A.~Loeliger, ``An introduction to factor graphs,'' IEEE Sig. Proc. Magazine, vol. 21, no. 1, pp. 28--41, January 2004.}

\mybibitem{FeWa_05}{J.~Feldman, M.~J. Wainwright, and D.~R. Karger, ``Using linear programming to decode binary linear codes,'' IEEE Trans. Inf. Theory, vol. 51, no. 3, pp. 954--972, March 2005.}

\mybibitem{VoKo_06}{P.~O.~Vontobel and R.~Koetter, ``Towards low-complexity linear-programming decoding,'' in Proceedings of 4th Int. Conf. on Turbo Codes and Related Topics, Munich, Germany, April 3--7, 2006.}

\mybibitem{VoKo_07}{P. O. Vontobel and R. Koetter, ``On low--complexity linear-programming decoding of LDPC codes,'' European Transactions on Telecommunications, vol. 18, no. 5, pp. 509--517, April 2007.}

\mybibitem{GlJa_07}{A.~Globerson and T.~Jaakkola, ``Fixing max-product: convergent message passing algorithms for MAP LP-relaxations,'' In Advances in Neural Information Processing Systems, pp. 553--560, 2007.}

\mybibitem{FlSk_09}{M.~F.~Flanagan, V.~Skachek, E.~Byrne, and M.~Greferath, ``Linear-programming decoding of nonbinary linear codes,'' IEEE Trans. Inf. Theory, vol. 55, no. 9, pp. 4134--4154, September 2009.}

\mybibitem{GoBu_12}{D.~Goldin and D.~Burshtein, ``Iterative linear programming decoding of non-binary linear codes with linear complexity,'' IEEE. Trans. Inf. Theory, vol. 59, no. 1, pp. 282--300, January 2013.}

\mybibitem{MaHa_09}{Y.~Maeda and K.~Haruhiko, ``Error control coding for multilevel cell flash memories using nonbinary low-density parity-check codes,'' 24th IEEE Int. Symp. on Defect and Fault Tolerance in VLSI Systems, pp.~367--375, Chicago, IL, USA, October 7--9, 2009.}

\mybibitem{Bu_09}{D.~Burshtein, ``Iterative approximate linear programming decoding of LDPC codes with linear complexity,'' IEEE Trans. Inf. Theory, vol. 55, no. 11, pp. 4835--4859, November 2009.}

\mybibitem{AnKa_10}{I.~Andriyanova and K.~Kasai, ``Finite-length scaling of non-binary $(c, d)$ LDPC Codes for the BEC,'' 2010 IEEE Int. Symp. on Inf. Theory, pp.~714--718, Dallas, TX, USA, June 13--18, 2010.}

\mybibitem{VoKo_IT}{P.~O.~Vontobel and R.~Koetter, ``Graph-cover decoding and finite-length analysis of message-passing iterative decoding of LDPC codes,'' CoRR,  \texttt{http://www.arxiv.org/abs/cs.IT/0512078}, December~2005.}

\mybibitem{YaFe_06}{K.~Yang, J.~Feldman, and X.~Wang ``Nonlinear programming approaches to decoding low-density parity-check codes,'' IEEE Journal on Selected Areas in Communications, vol. 24, no. 8, pp. 1603--1613, August 2006.}

\mybibitem{Be_99}{D.~Bertsekas, {\it Nonlinear Programming}. Belmont, MA: Athena Scientific, second ed., 1999.}


\mybibitem{YaWa_08}{K.~Yang, X.~Wang, and J.~Feldman, ``A new linear programming approach to decoding linear block codes,'' IEEE Trans. Inf. Theory, vol. 54, no. 3, pp. 1061--1072, March 2008.}

\mybibitem{TaSh_11}{M.~H.~ Taghavi, A.~ Shokrollahi, and P.~H.~Siegel, ``Efficient implementation of linear programming decoding,'' IEEE Trans. Inf. Theory, vol. 57, no. 9, pp. 5960--5982, September 2011.}

\mybibitem{BaLu_08}{S.~Barman, X.~Liu, S.~Draper, and B.~Recht, ``Decomposition methods for large scale LP decoding,'' in Proceedings of 49th Annual Allerton Conference on Communication, Control, and Computing, pp. 253--260, September 28--30, 2011.}

\mybibitem{GoBu_10}{D.~Goldin and D.~Burshtein, ``Approximate iterative LP decoding of LDPC codes over GF(q) in linear complexity,'' in Proceedings of 26th Convention of Electrical and Electronics Engineers in Israel, pp. 960--964, November 17--20, 2010.}

\mybibitem{VoSh_09}{P.~O.~Vontobel and S.~Jalali, ``Coordinate-ascent method for linear programming decoding,'' U.S. Patent Application 11 831 716, February 5, 2009.}

\mybibitem{Pu_12}{M.~Punekar, {\it Efficient LP Decoding of Binary and Nonbinary Linear Codes}, Ph.D. Thesis, School of Electrical, Electronic and Communications Engineering, University College Dublin, Dublin, Ireland, June 2012.}

\mybibitem{PuFl_10}{M.~Punekar and M.~F.~Flanagan, ``Low-complexity LP decoding of nonbinary linear codes,'' in Proceedings of 48th Annual Allerton Conference on Communication, Control, and Computing, pp.~6--13, September 29 -- October 1, 2010.}

\mybibitem{PuFl_11}{M.~Punekar and M.~F.~Flanagan, ``Trellis-based check node processing for low-complexity nonbinary LP decoding,'' in Proceedings of 2011 IEEE Int. Symp. on Information Theory, St. Petersburg, Russia, pp. 1653--1657, August 1--5, 2011.}

\mybibitem{DeFo_07}{D.~Declercq and  M.~Fossorier, ``Decoding algorithms for nonbinary LDPC codes over GF(q),'' IEEE Trans. Comm., vol. 55, no. 4, pp. 633--643, April 2007.}

\end{document}